\documentclass[12pt]{article}

\usepackage[T1]{fontenc}
\usepackage{lmodern}

\usepackage{setspace}
\usepackage[margin=1.25in]{geometry}

\usepackage[dvipsnames]{xcolor}
\usepackage{pdfpages}
\usepackage{float}
\usepackage{multirow}
\usepackage{caption}
\usepackage{subcaption}
\usepackage{epigraph}

\usepackage{mathtools}
\usepackage{amssymb}
\usepackage{amsthm}
\usepackage{amsfonts}
\usepackage{aliascnt}
\usepackage{accents}
\usepackage{dutchcal}
\usepackage{bbm}
\usepackage{mathrsfs}

\usepackage{enumitem}
\usepackage{sgame}
\usepackage{tikz}
\usepackage{tikz-cd}
\usetikzlibrary{calc,shapes,arrows}
\usepackage{tcolorbox}
\usepackage{listings}
\usepackage[normalem]{ulem}

\usepackage[round]{natbib}
\DeclareMathOperator{\supp}{supp}

\DeclareMathOperator{\conv}{conv}

\newcommand{\E}{\mathbb{E}}
\newcommand{\R}{\mathbb{R}}
\newcommand{\Q}{\mathbb{Q}}

\DeclareMathOperator*{\esssup}{ess\,sup}
\DeclareMathOperator*{\essinf}{ess\,inf}

\newcommand{\Law}{\mathcal L}
\newcommand{\one}{\mathbf 1}

\newcommand{\cat}{\mathrm{cat}}

\newtheorem{theorem}{Theorem}[section]

\newaliascnt{proposition}{theorem}
\newtheorem{proposition}[proposition]{Proposition}
\aliascntresetthe{proposition}

\newaliascnt{lemma}{theorem}
\newtheorem{lemma}[lemma]{Lemma}
\aliascntresetthe{lemma}

\newaliascnt{corollary}{theorem}
\newtheorem{corollary}[corollary]{Corollary}
\aliascntresetthe{corollary}

\newaliascnt{claim}{theorem}
\newtheorem{claim}[claim]{Claim}
\aliascntresetthe{claim}

\theoremstyle{definition}

\newaliascnt{definition}{theorem}
\newtheorem{definition}[definition]{Definition}
\aliascntresetthe{definition}

\newaliascnt{example}{theorem}

\aliascntresetthe{example}

\newaliascnt{assumption}{theorem}

\aliascntresetthe{assumption}

\newaliascnt{condition}{theorem}

\aliascntresetthe{condition}

\newaliascnt{question}{theorem}

\aliascntresetthe{question}

\newaliascnt{remark}{theorem}
\newtheorem{remark}[remark]{Remark}
\aliascntresetthe{remark}

\newaliascnt{remarks}{theorem}

\aliascntresetthe{remarks}

\newaliascnt{aside}{theorem}

\aliascntresetthe{aside}

\newaliascnt{note}{theorem}

\aliascntresetthe{note}

\usepackage{thmtools}
\usepackage{thm-restate}

\usepackage{hyperref}

\hypersetup{
    colorlinks=true,
    linkcolor=OrangeRed,
    filecolor=Thistle,
    urlcolor=Thistle,
    citecolor=Thistle,
}

\usepackage[nameinlink]{cleveref}

\crefname{theorem}{theorem}{theorems}
\Crefname{theorem}{Theorem}{Theorems}
\crefname{proposition}{proposition}{propositions}
\Crefname{proposition}{Proposition}{Propositions}
\crefname{lemma}{lemma}{lemmas}
\Crefname{lemma}{Lemma}{Lemmas}
\crefname{corollary}{corollary}{corollaries}
\Crefname{corollary}{Corollary}{Corollaries}
\crefname{claim}{claim}{claims}
\Crefname{claim}{Claim}{Claims}
\crefname{definition}{definition}{definitions}
\Crefname{definition}{Definition}{Definitions}
\crefname{example}{example}{examples}
\Crefname{example}{Example}{Examples}
\crefname{assumption}{assumption}{assumptions}
\Crefname{assumption}{Assumption}{Assumptions}

\makeatletter
\let\cref@old@isrefconsecutive\cref@isrefconsecutive
\def\cref@isrefconsecutive#1#2{%
  \begingroup
    \def\cref@assumptiontype{assumption}%
    \cref@gettype{#1}{\cref@typea}%
    \ifx\cref@typea\cref@assumptiontype
      \endgroup
      \@cref@refconsecutivefalse
    \else
      \endgroup
      \cref@old@isrefconsecutive{#1}{#2}%
    \fi
}
\makeatother

\crefname{condition}{condition}{conditions}
\Crefname{condition}{Condition}{Conditions}
\crefname{question}{question}{questions}
\Crefname{question}{Question}{Questions}
\crefname{remark}{remark}{remarks}
\Crefname{remark}{Remark}{Remarks}
\crefname{remarks}{remarks}{remarks}
\Crefname{remarks}{Remarks}{Remarks}
\crefname{aside}{aside}{asides}
\Crefname{aside}{Aside}{Asides}
\crefname{note}{note}{notes}
\Crefname{note}{Note}{Notes}
\crefname{appendix}{appendix}{appendices}
\Crefname{appendix}{Appendix}{Appendices}

\newcommand{\secref}[1]{\hyperref[#1]{\S\ref*{#1}}}

\definecolor{backcolour}{rgb}{0.63,0.79,0.95}
\lstdefinestyle{mystyle}{
  backgroundcolor=\color{backcolour},
  basicstyle=\ttfamily\footnotesize,
  breakatwhitespace=false,
  breaklines=true,
  captionpos=b,
  keepspaces=true,
  numbers=left,
  numbersep=5pt,
  showspaces=false,
  showstringspaces=false,
  showtabs=false,
  tabsize=2
}
\begin{document}
\title{Multidimensional Risk Made Easy}
\author{Mark Whitmeyer\thanks{Arizona State University. Email: \href{mailto:mark.whitmeyer@gmail.com}{mark.whitmeyer@gmail.com}. Dedicated to KS. I used ChatGPT as one would an RA and \href{https://www.refine.ink}{refine.ink} for feedback.}}
\date{\today}
\maketitle

\begin{abstract}
Suppose we want to assign a certainty equivalent--one number--to a multivariate risk. Which such assignments are law-invariant, monotone with respect to vector stochastic dominance, and invariant to independent background risk? I show that every such certainty equivalent is a positive mixture of scalar entropic certainty equivalents applied to positive projections of the vector risk. The same representation yields a robust-order characterization: unanimity across such certainty equivalents is equivalent, up to closure, to dominance after adding independent multidimensional background risk. In a social-welfare specialization, the corresponding shadow valuations are welfare weights.
\end{abstract}

\section{Introduction}\label{sec:intro}

In this paper, I characterize a natural class of certainty equivalents for vector-valued risks. In one dimension, the operation looks innocent enough: a random number is replaced by a sure number, and nothing has had to say where that number lives. With vector-valued risks, this innocence is gone. There is a small preliminary nuisance that one dimension hides: a certainty equivalent is supposed to be a sure payoff replacing the risk, but here the risk is a vector. If that replacement is to be summarized by a number, the number has to be a number \textit{of something}. 

I fix that something by choosing a \textit{numeraire} direction \(e\), and I measure sure equivalents along the line \(\R e\). For a risk \(X\), the sure equivalent payoff is \(\Phi(X)e\). The scalar \(\Phi(X)\) is the certainty equivalent: the number of numeraire units in that payoff so that \(\Phi(X)=c\) means indifference between risk \(X\) and sure thing \(ce\).

Importantly, \(\Phi(X)\) is a scalar even though \(X\) is multidimensional. This is the little sleight of hand that the notation makes look harmless. One number is being asked to stand in for a random vector, and the fact that the answer is a number does not by itself say what the number is allowed to ignore. It may forget the bookkeeping by which the uncertainty is described, but not the distribution of the risk. It should respect directions that are unambiguously better, but not pretend that all vector movements are already comparable. It ought to be stable in the presence of unrelated independent noise, but not let such noise rewrite the ranking of the incremental risks being compared. I turn these requirements into three benchmark axioms. The certainty equivalent must be law-invariant, monotone with respect to vector first-order stochastic dominance, and invariant to independent background risk. Once normalized by the numeraire, these are what I term the \textit{additive certainty equivalents}.

The last of these desiderata says that the ranking of two incremental vector risks should not change merely because an unrelated independent risk is also present. For certainty equivalents, this is exactly additivity across independent risks. The monotonicity requirement needs an order on vectors. I encode this order by a set \(C\) of unambiguous improvements: directions in which moving is, by assumption, better. \textit{Viz.}, \(x+c\) is better than \(x\) whenever \(c\in C\). In the usual distributional example, \(C=\R^d_+\), so an improvement is just more for at least one group and less for none. I use the more abstract notation because the same argument applies to any closed convex cone of such improvements.

This discipline has bite even when nothing is risky. For a certainty equivalent \(\Phi\), I write \(\ell(x)\coloneqq \Phi(x)\) (\(x\in V\)) for the value it assigns to a deterministic payoff. Additivity makes \(\ell\) additive. Monotonicity rules out pathologies, makes \(\ell\) linear, and also makes \(\ell\) positive on the improvement cone: if \(c\in C\), then \(\ell(c)\ge0\). Since certainty equivalents are measured in units of the numeraire direction \(e\), I normalize by \(\ell(e)=1\).

The same order and numeraire also determine a larger family of positive linear tests. These are the linear functionals \(q\in C^*\) with \(q(e)=1\). Equivalently, a \(q\) assigns nonnegative value to every unambiguous improvement and assigns value one to the numeraire. I call these tests \textit{shadow valuations}, meant literally. That is, I am not assuming that the decision maker begins with a little linear price system in her head and then evaluates risk through it. The \(q\)s are the linear objects the representation forces into view.

What a shadow valuation does is simple. Applied to a deterministic payoff \(x\in V\), it produces the scalar \(q\cdot x\). Applied to a vector risk \(X\), it yields the real random variable \(q\cdot X\). Keep in mind that at this point I have said nothing about risk attitudes, but only described the transformations \(q\cdot X\) through which a vector risk can be turned into an ordinary real risk. The natural objection is: why should there be many such readings? Deterministic choice has already produced one, namely \(\ell\). Why not just evaluate \(\ell\cdot X\) and be done?

This is the tempting shortcut, and it is exactly where multidimensional risk reappears. Let the risk be two dimensional \(V=\R^2\), let the numeraire be \(e=(1,1)\), let improvements be coordinatewise, and suppose the deterministic valuation is \(\ell=(1/2,1/2)\). Now compare
\[
  X=
  \begin{cases}
    (1,0), &\text{with probability }1/2,\\
    (0,1), &\text{with probability }1/2,
  \end{cases}
  \qquad\text{and}\qquad
  Y=(1/2,1/2)\quad\text{for sure.}
\]
Under the deterministic evaluation \(\ell\), \(\ell\cdot X=\ell\cdot Y=1/2\) almost surely. But look at what \(\ell\cdot X\) has forgotten. A planner may care that, \textit{ex post}, one group receives everything. An investor may care which account or currency pays off in which state. A household may care whether wage risk and health risk arrive together or separately. These are things the deterministic tradeoff was never asked, and is not rich enough, to answer.

My theorem says what replaces the shortcut. The first ingredient is familiar from \citet{MPST2024}: entropic risk sensitivity. The second is specific to vector risk: the shadow valuation \(q\), which determines which real payoff \(q\cdot X\) is being evaluated. My main result, \Cref{thm:scalar-main}, says that every additive certainty equivalent averages scalar entropic evaluations of these shadow-value readings: \(\Phi(X) = \int K_a(q\cdot X) dm(a,q)\), where \(K_a(Z)\) equals \(\frac1a\log\E \mathrm e^{aZ}\) if \(a\in\R\setminus\{0\}\) and \(\E[Z]\) if \(a = 0\).\footnote{If \(a = + \infty\), \(K_a(Z) = \esssup Z\); and if \(a = - \infty\), \(K_a(Z) = \essinf Z\).} The parameter \(a\) describes risk sensitivity along the real risk \(q\cdot X\). The shadow valuation \(q\) describes the manner of aggregation. 

Moreover, when deterministic payoffs are evaluated by \(\ell\), the representing measure satisfies \(\int q dm=\ell\). Accordingly, deterministic choices identify the average shadow valuation, but not the distribution of shadow valuations around that average; they specify the center of mass, whereas risky multidimensional choices reveal what is spread around it. All in all, a multidimensional risk attitude decomposes into jointly distributed shadow valuations and scalar risk sensitivities.\footnote{The main theorem of \citet{MPST2024} starts with a real-valued risk, so the representation has only the risk-sensitivity parameter \(a\). Here the risk is vector-valued. The scalar entropic certainty equivalents remain the building blocks, while the projection \(q\cdot X\) of the vector risk becomes part of the representation. Applying their scalar theorem after the deterministic projection \(\ell\cdot X\) is the special case in which the representing measure is concentrated at \(q=\ell\).}

I close with two applications. In a social-welfare application with \(C=\R^d_+\) and \(e=(1,\ldots,1)\), shadow valuations are welfare weights. I fix deterministic weights \(\ell\) for a Pareto-monotone planner and ask which local risky comparisons survive all risk-averse representations with those weights. I find that every welfare-weighted aggregate compatible with the support of \(\ell\) must be less variable.

My second application is a comparison of income streams. The numeraire is a benchmark stream, and a shadow valuation \(q\) is a normalized positive valuation of dated payoffs in units of that benchmark. I show that multiple-discount-rate unanimity has a risky analog. The deterministic test asks whether \(q\cdot x\ge q\cdot y\) for every relevant dated valuation \(q\). The risky test asks whether \(K_a(q\cdot X)\ge K_a(q\cdot Y)\), so the comparison must survive both the dated valuation \(q\) and the scalar risk sensitivity \(a\). When the streams are deterministic, the test collapses back to present value.

\smallskip

\noindent \textbf{Roadmap.}
\secref{sec:model} introduces the environment, before I prove the representation theorem in \secref{sec:representation}. \secref{sec:robust} characterizes robust comparisons and their background-risk interpretation. \secref{sec:apps} discusses applications to distributional social welfare and comparisons of income streams over time. Formal proofs reside in \Cref{app:proofs}.

\subsection{Related literature}

This paper belongs to the literature on risk with many commodities. The classical expected-utility approach fixes a utility function on vector outcomes and asks how its curvature,\footnote{See, e.g., \citet{Stiglitz1969}, \citet{KihlstromMirman1974}, and \citet{KihlstromMirman1981}.} cross-curvature,\footnote{For instance, \citet{Richard1975}.} and induced risk premia\footnote{E.g., \citet{Duncan1977}, \citet{Karni1979}, and \citet{LevyLevy1991}.} describe attitudes toward multivariate risk. That program produces multivariate conceptions of risk aversion and comparative risk aversion.\footnote{Refer to \citet{Keeney1973}, \citet{Schlee1990}, and \citet{Grant1995}.} I study the same fundamental object--a random vector--but not the same representation problem.

A related branch studies how risks are paired across attributes. One strand provides dominance and derivative-sign conditions for multivariate utility functions;\footnote{For example, \citet{Scarsini1988} and \citet{EeckhoudtReySchlesinger2007}.} another studies risk apportionment and the preference for combining good outcomes in some dimensions with bad outcomes in others.\footnote{See \citet{EeckhoudtSchlesinger2006}, \citet{EeckhoudtSchlesingerTsetlin2009}, and \citet{TsetlinWinkler2009}.} Recent work also decomposes multidimensional risk premia into the pieces generated by the certainty-indifference map and the curvature of utility under risk \citep{EeckhoudtPaganiPeluso2023}. My paper is closest to that line when two risks have the same deterministic scalarization but differ in how payoffs are distributed across dimensions and states. 

There is also a non-expected-utility and multivariate risk-measure literature. Some papers extend generalized expected utility and dominance axioms to multidimensional distributions;\footnote{For instance, \citet{Karni1989} and \citet{SafraSegal1993}.} others extend dual theory, local utility, or coherent risk measurement to multivariate risks.\footnote{\citet{GalichonHenry2012}, \citet{EkelandGalichonHenry2012}, and \citet{CharpentierGalichonHenry2016}, in turn. For set-valued risk measures, reference \citet{HamelHeyde2010} and \citet{HamelHeydeRudloff2011}.} 
The closest recent behavioral comparator is \citet{KeZhang2025}, who ask how a decision maker organizes multidimensional risk: whether dimensions are aggregated before risk is evaluated, whether risk is evaluated dimension by dimension before aggregation, or whether dimensions are evaluated recursively in an endogenous order. My question and theirs are complementary.

The mathematical antecedent of my work is the one-dimensional theory of monotone additive statistics.\footnote{The one-dimensional result is \citet{MPST2024}. See also \citet{MuPomattoStrackTamuz2024Background} and especially \citet{Fritz2024}, whose technical apparatus is fundamental to my exercise.} I show that the entropic scalar blocks from that theory persist. The link is the shadow-valuation profile: every positive valuation \(q\) turns the vector risk into a scalar risk, and every scalar risk-sensitivity parameter evaluates that projection.

\section{Preliminaries, Mathematical and Economic}\label{sec:model}

Throughout, let \(V=\R^d\), let \(C\subseteq V\) be a closed, convex, pointed cone with nonempty interior, and fix a numeraire direction \(e\in\operatorname{int}C\).\footnote{For a set \(X \subseteq V\), \(\operatorname{int} X\) denotes the topological interior of \(X\). The choice of norm does not matter.} E.g., \(C=\R^d_+\) and \(e=(1,\ldots,1)\).

I write \(x\le_C y\) if \(y-x\in C\). For \(a,b\in V\), \(\left[a,b\right]\coloneqq\left\{v\in V\colon a\le_C v\le_C b\right\}\) is the corresponding order interval. The positive dual cone is
\[
  C^*\coloneqq\left\{q\in V^*\colon q(c)\ge0\text{ for every }c\in C\right\},
\]
where \(V^*\) is the space of linear functionals on \(V\). Since \(V=\R^d\), I identify \(q\in V^*\) with its coefficient vector and write \(q\cdot x\) for \(q(x)\). Define
\[Q\coloneqq\left\{q\in C^*\colon q(e)=1\right\}.\]
Elements of \(Q\) are positive shadow valuations, normalized so that one unit of the numeraire direction has value one. Mathematically, \(Q\) is the base of the dual cone \(C^*\) selected by the normalization \(q(e)=1\).

Fix any norm \(\|\cdot\|\) on \(V\), and let \(\|\cdot\|_*\) be the dual norm on \(V^*\).

\begin{lemma}\label[lemma]{lem:dual-base-order-unit}
The set \(Q\) is nonempty. Every \(q\in C^*\setminus\left\{0\right\}\) satisfies \(q(e)>0\), and
\[C^*=\left\{\lambda q\colon \lambda\ge0,\ q\in Q\right\}.
\]
Moreover, \(\operatorname{int}C^*\ne\varnothing\), and, for every \(v\in V\), there is \(M>0\) such that \(-Me\le_C v\le_C Me\). The order interval \(\left[-e,e\right]\) contains a neighborhood of \(0\) and so every bounded subset of \(V\) is contained in some order interval \(\left[-Me,Me\right]\).
\end{lemma}

\begin{lemma}\label[lemma]{lem:Q-compact}
The set \(Q\) is compact and convex.
\end{lemma}

Economically, we have fixed the language in which we will compare multidimensional risks. The space \(V\) is the list of payoff dimensions: goods, dates, accounts, groups, or other attributes. The cone \(C\) says which deterministic changes are unambiguous improvements, and the order \(x \leq_C y\) means that moving from \(x\) to y is better in every direction recognized by that cone. The numeraire direction \(e\) is the unit in which we measure certainty equivalents. For example, for social welfare, if the numeraire is \(e=(1,\dots,1)\), then \(c\) means a uniform transfer of \(c\) to every group. For consumption over time, \(e\) might be a constant consumption stream. For portfolios, \(e\) might be the cash account or benchmark payoff. The dual cone \(C^*\) collects all linear valuations that respect the improvement order, and \(Q\) normalizes these valuations so that one unit of \(e\) has value one. Consequently, an element \(q \in Q\) is a possible system of shadow prices, welfare weights, or marginal values across dimensions.

Let \(L^\infty(V)\) denote the essentially bounded \(V\)-valued random variables, identified up to almost-sure equality. For \(X\in L^\infty(V)\), I write \(\Law(X)\) for its law, the probability measure on \(V\) induced by \(X\). Thus, for every Borel set \(A\subseteq V\), \(\Law(X)(A)=\mathbb P(X\in A)\).

All random variables live on a probability space rich enough that, whenever \(X_1,\ldots,X_n\in L^\infty(V)\) and \(\nu\) is a compactly supported Borel probability law on \(V\), there is a \(\nu\)-distributed random vector independent of \(X_1,\ldots,X_n\). Because all functionals below are law-invariant, the particular probability
space is only a device for realizing independent random variables. Equivalently, I could work directly with compactly supported laws on \(V\). Under this law-level interpretation, the sum of independent random variables with laws \(\mu\) and \(\nu\) has law \(\mu*\nu\), their convolution.

\begin{definition}Let \(\mu,\nu\) be compactly supported Borel probability measures on \(V\). Write \(\mu\le_{\mathrm{st}}\nu\)
if \(\mu(U)\le \nu(U)\) for every closed \(C\)-upper set \(U\), i.e., every closed set satisfying \(U+C\subseteq U\).

For bounded random variables \(X,Y\), write \(X\succeq_{\mathrm{st}} Y\) if \(X\) (first-order) stochastically dominates \(Y\): \(\Law(Y)\le_{\mathrm{st}}\Law(X)\).
\end{definition}

\begin{remark}
For compactly supported laws on finite-dimensional ordered spaces, this stochastic order agrees with the usual testing order against bounded continuous \(C\)-increasing functions. I use the closed-upper-set formulation to match the stochastic preorder used by \citet[Proposition 4.1 and Theorem 5.6]{Fritz2024}.
\end{remark}

A vector risk \(X\) is a random payoff vector; and its law records only the distribution of those payoffs, not the particular probability space used to realize them. For my purposes, vector stochastic dominance is the natural monotone order on such laws: \(X\) dominates \(Y\) when \(X\) places more probability on upper outcomes, where ``upper'' is defined by \(C\). The rest of the paper asks how certainty equivalents can evaluate these vector risks while respecting this order and the normalization imposed by \(e\).

\subsection{Certainty equivalents and background risk}

I assume a complete, reflexive, and transitive preference relation \(\succeq\) on \(L^\infty(V)\), writing \(\sim\) for its symmetric part. I say that a \textit{preference has an \(e\)-certainty equivalent} if for every \(X\in L^\infty(V)\) there is a unique number \(c\in\R\) such that \(X\sim ce\) and write this number as \(\Phi_\succeq(X)=c\). The \textit{\(e\)-certainty equivalent} \(\Phi_\succeq(X)\) is the amount of the numeraire direction \(e\) that is equivalent to the vector risk \(X\).

I say that \(\succeq\) is \textit{law-invariant} if \(X\stackrel{d}{=}Y\) implies \(X\sim Y\). It is \textit{cone-monotone} if
\[X\succeq_{\mathrm{st}}Y
  \quad\Longrightarrow\quad
  X\succeq Y.\]
It is \textit{invariant to independent background risk} if, whenever \(Z \in L^\infty(V)\) is independent of the pair \((X,Y)\),
\[X\succeq Y
  \quad\iff\quad
  X+Z\succeq Y+Z.\]

\begin{definition}
Let \(\mathscr E\) denote the class of complete, reflexive, and transitive preferences on \(L^\infty(V)\) that are law-invariant, cone-monotone, invariant to independent background risk, and admit \(e\)-certainty equivalents. I term an element of \(\mathscr E\) an \textit{admissible} preference.
\end{definition}

\begin{definition}
\label{def:scalar-mas}
An \textit{additive certainty functional (ACF)} is a map \(\Phi \colon L^\infty(V)\to\R\) such that:
\begin{enumerate}[label=(\roman*),noitemsep]
  \item \(\Phi\) is law-invariant;
  \item \(X\succeq_{\mathrm{st}}Y\) \(\Longrightarrow\) \(\Phi(X)\ge\Phi(Y)\);
  \item \(X,Y\) are independent \(\Longrightarrow\) \(\Phi(X+Y)=\Phi(X)+\Phi(Y)\).
\end{enumerate}
It is an \textit{additive certainty equivalent (ACE)} if \(\Phi(ce)=c\) for all \(c \in \R\).
\end{definition}

Economically, this subsection has moved us from rankings of vector risks to numerical certainty equivalents measured in the numeraire direction \(e\). The fundamental object is still a preference relation over multidimensional risks, but the existence of an \(e\)-certainty equivalent lets each risk \(X\) be summarized by the amount of the numeraire \(\Phi_\succeq(X)\), or, equivalently, by the deterministic vector \(\Phi_\succeq(X)e\), that is judged indifferent to it. The definition of an ACF specifies the properties that such a numerical summary should inherit from the underlying preference: law-invariance says that only the distribution of payoffs matters, cone-monotonicity says that vector stochastic improvements cannot lower the certainty equivalent, and additivity says that independent risks contribute separately. 

This last property is the functional counterpart of invariance to independent background risk. If adding the same independent risk to both sides never changes a comparison, then the numerical value of that independent risk must enter both sides as the same additive term. The natural question is whether replacing preferences by such functionals loses any behavioral content, or imposes anything beyond the assumptions already stated. The next proposition answers no. Once we normalize certainty equivalents by \(\Phi(ce)=c\), admissible preferences and ACEs are equivalent descriptions of the same object.

\begin{proposition}
\label[proposition]{prop:preference-mas-bridge}
The map \(\succeq\mapsto \Phi_\succeq\) is a one-to-one correspondence between \(\mathscr E\) and ACEs. The inverse sends an ACE \(\Phi\) to the preference \(\succeq_\Phi\) defined by \(X\succeq_\Phi Y\) \(\iff\) \(\Phi(X)\ge\Phi(Y)\).
\end{proposition}
Accordingly, admissible preferences and ACEs are two equivalent descriptions of the same normalized objects.

Having reduced admissible preferences to additive certainty equivalents, I now introduce the (scalar) evaluations that will generate all of them. A vector risk can be viewed through any normalized positive valuation \(q \in Q\). This valuation turns the multidimensional payoff \(X\) into the scalar risk \(q \cdot X\): the payoff measured in the units and tradeoffs encoded by \(q\). Once this scalar projection is fixed, we can evaluate its risk by a one-dimensional entropic certainty equivalent. The parameter \(a\) indexes scalar risk sensitivity along that valuation direction: negative \(a\) corresponds to risk aversion, \(a=0\) to risk neutrality, positive \(a\) to risk seeking, and the endpoints \(a=\pm\infty\) to worst- and best-case evaluations.

The profile I define below captures all of these scalar views of \(X\). For each pair \((a,q)\), it asks what certainty equivalent we obtain by first projecting the vector risk using \(q\) and then applying the scalar entropic evaluator indexed by \(a\). Thus, \(q\) captures the shadow valuation used to \textit{aggregate} dimensions, while \(a\) captures \textit{risk sensitivity along that aggregate}. The representation theorem will show that this profile is exhaustive: every law-invariant, cone-monotone, additive certainty equivalent is a positive mixture of these scalar entropic projections. The elementary properties established in this subsection--continuity, Lipschitzness, additivity under independent sums, and agreement with \(q \cdot x\) on deterministic vectors--are the facts needed to make that representation possible.

For a bounded real random variable \(Z\), define
\[
K_a(Z)=
\begin{cases}
\frac1a\log\E \mathrm e^{aZ}, \quad &\text{if} \quad a\in\R\setminus\{0\},\\
\E[Z], \quad &\text{if} \quad a = 0,\\
\esssup Z, \quad &\text{if} \quad a =+\infty,\\
\essinf Z, \quad &\text{if} \quad a =-\infty.
\end{cases}
\]
Let \(\overline{\R}=[-\infty,+\infty]\) be the two-point compactification of \(\R\), and define \(\Theta=\overline{\R}\times Q\). For \(X\in L^\infty(V)\), define \(\kappa_X \colon \Theta\to\R\) by \(\kappa_X(a,q)=K_a(q\cdot X)\), and note that the random variable \(q\cdot X\) is bounded for every \(q\in Q\). Let \(C(\Theta)\) denote the space of continuous real-valued functions on \(\Theta\), equipped with the \(\sup\) norm. For any bounded \(f\colon\Theta\to\R\), write \(\left\|f\right\|_\infty\coloneqq\sup_{\theta\in\Theta}\left|f(\theta)\right|\).

Next, I define the \textit{order-unit norm} \(\left\|\cdot\right\|_e\) by
\[
\left\|v\right\|_e\coloneqq\inf\left\{r>0\colon -re\le_C v\le_C re\right\}.
\]
By \Cref{lem:dual-base-order-unit}, this number is finite for every \(v\in V\). It is a norm,\footnote{Symmetry follows from the symmetric definition of the order interval, and homogeneity and the triangle inequality follow from the fact that \(C\) is a cone. If \(\left\|v\right\|_e=0\), then, for every \(n\), there is \(r_n<1/n\) such that \(-r_ne\le_C v\le_C r_ne\). Hence,\(-e/n\le_C v\le_C e/n\). Equivalently, \(v+e/n\in C\) and \(e/n-v\in C\). Since \(C\) is closed, letting \(n\to\infty\) delivers \(v\in C\) and \(-v\in C\). Pointedness produces \(v=0\).} and since \(V\) is finite-dimensional, \(\left\|\cdot\right\|_e\) induces the usual topology.

For \(X\in L^\infty(V)\), write \(\left\|X\right\|_{\infty,e}\coloneqq \esssup \left\|X\right\|_e\). For \(X,Y\in L^\infty(V)\), write \(\left\|X-Y\right\|_{\infty,e}\coloneqq \esssup \left\|X-Y\right\|_e\).

I now prove the regularity properties that make the profile usable as an object of analysis. Its economic content is simple. The order-unit norm measures the size of a perturbation in units of the numeraire: saying that \(X\) and \(Y\) differ by at most \(r\) in this norm means that, state by state, the difference can be covered by \(r\) units of \(e\) in either direction. Since every \(q \in Q\) satisfies \(q(e)=1\), no normalized positive valuation can magnify such a perturbation by more than \(r\). The entropic certainty equivalents have the same \(1\)-Lipschitz property with respect to uniform changes in their scalar argument. Accordingly, the entire profile changes by at most \(r\), uniformly over all valuation directions \(q\) and all risk-sensitivity parameters \(a\). \Cref{lem:profile-lipschitz} also says that nearby valuation directions and nearby risk sensitivities provide nearby scalar evaluations, with the endpoint parameters \(a=\pm\infty\) serving as the limiting worst- and best-case evaluations. Thus, each vector risk \(X\) generates a bounded continuous function \(\kappa_X\) on the compact parameter space \(\Theta\).

\begin{lemma}\label[lemma]{lem:profile-lipschitz}
For every \(X\in L^\infty(V)\), \(\kappa_X\) is bounded and belongs to \(C(\Theta)\). Moreover, for every \(X,Y\in L^\infty(V)\),
\(\left\|\kappa_X-\kappa_Y\right\|_\infty\le\left\|X-Y\right\|_{\infty,e}\).
\end{lemma}

The second nice property is compatibility with independent background risk. For a fixed valuation \(q\), independent vector risks \(X\) and \(Y\) become independent \textit{scalar} risks \(q \cdot X\) and \(q \cdot Y\). Entropic certainty equivalents are additive over independent scalar sums, so the profile of an independent sum is the sum of the profiles. This is the reason entropic scalar evaluations are the right elementary building blocks for additive certainty equivalents. The deterministic case is equally important. If \(x\) is non-random, then \(q \cdot x\) is a constant, and every scalar risk-sensitivity parameter \(a\) assigns that constant the same value. Thus, on deterministic outcomes, the profile remembers the valuation direction \(q\) but not the risk-sensitivity parameter \(a\). This fact will later become the barycenter condition linking the representing measure to deterministic shadow valuations.

\begin{lemma}\label[lemma]{lem:profile-additivity}
If \(X,Y\in L^\infty(V)\) are independent, then \(\kappa_{X+Y}=\kappa_X+\kappa_Y\). For deterministic \(x\in V\), \(\kappa_x(a,q)=q\cdot x\).
\end{lemma}

Together, \Cref{lem:profile-lipschitz,lem:profile-additivity} show that the entropic projection profile has the same formal features as the certainty equivalents we seek to represent. It is stable as a continuous function of the risk, additive under independent sums, and on deterministic vectors it reduces to ordinary valuation by \(q\). The representation theorem in the next section uses these facts in the reverse direction: every law-invariant, cone-monotone, additive certainty equivalent can be viewed as a positive linear functional of these profiles, and, consequently, as a positive mixture of their scalar coordinates.

\section{The Main Representation}\label{sec:representation}

In the previous section, I defined the profile \((a,q)\mapsto K_a(q\cdot X)\), associating each vector risk \(X\) with its entropic projection profile. This profile captures every scalar certainty equivalent obtained by first valuing the vector payoff through a normalized positive valuation \(q\in Q\), before then applying scalar risk sensitivity \(a\) to the resulting scalar risk \(q\cdot X\). The regularity and additivity properties established above make these profiles suitable objects of analysis: \(\kappa_X\) is a continuous function on the compact parameter space \(\Theta\), it changes continuously with \(X\), it is additive under independent sums, and on deterministic vectors it reduces to ordinary valuation by \(q\).

The purpose of this section is to show that the profile is exhaustive. A law-invariant, cone-monotone certainty equivalent that is additive across independent risks has no additional degrees of freedom beyond positive mixtures of these scalar entropic projections. Equivalently, every such functional can be written as
\(\Phi(X)=\int_\Theta K_a(q\cdot X) dm(a,q)\) for some finite positive measure \(m\) on \(\Theta\). A draw \((a,q)\) from this representing measure has a simple as-if interpretation: the evaluator uses \(q\) to aggregate the dimensions of \(X\), evaluates risk in that scalar aggregate using \(K_a\), and then averages across such scalar evaluations.

The representation also separates deterministic tradeoffs from risky multidimensional evaluation. If \(\ell(x)=\Phi(x)\) denotes the deterministic restriction of \(\Phi\), then deterministic outcomes can identify only the average valuation direction in the mixture. Indeed, for deterministic \(x\), the scalar risk \(q\cdot x\) is constant, so every risk-sensitivity parameter \(a\) assigns it the same value. Thus the representing measure must satisfy the barycenter condition \(\int_\Theta q dm(a,q)=\ell\). Intuitively, deterministic choice pins down \(\ell\), while risky multidimensional choice may depend on how that deterministic valuation is decomposed into valuation directions \(q\) and scalar risk sensitivities \(a\). I first record the deterministic restriction formally before stating the representation theorem.

\subsection{Deterministic restrictions}

I state the representation theorem for ACFs rather than only for normalized ACEs, so I first separate the deterministic scale of a functional from its risky part. On deterministic vectors there is no risk to evaluate: the functional simply assigns values to certain multidimensional payoffs. Additivity over independent sums becomes additivity across certain vectors, and cone-monotonicity says that movements in the improvement cone cannot reduce value. Thus, the deterministic restriction of an ACF should be an ordinary positive valuation on \(V\).

For an ACF \(\Phi\) and deterministic \(x\in V\), define the ACF's deterministic restriction as \(\ell(x)=\Phi(x)\).
\begin{lemma}\label[lemma]{lem:detlin}
The deterministic restriction \(\ell\) is a positive linear functional, i.e., \(\ell\in C^*\).
\end{lemma}

\Cref{lem:detlin} states the deterministic benchmark that the mixture representation must match. For normalized certainty equivalents, \(\ell\in Q\) is the shadow valuation revealed by choices among sure vectors. For an unnormalized ACF, \(\ell(e)\) is the scale of the functional; in the representation below, it will become the total mass of the representing measure.

This also explains (getting ahead of ourselves slightly) why the representation theorem contains a barycenter condition. If a measure \(m\) represents \(\Phi\), then a deterministic vector \(x\) is evaluated as
\[
    \int_\Theta K_a(q\cdot x) dm(a,q)
    =
    \int_\Theta q\cdot x dm(a,q)
    =
    \left(\int_\Theta q dm(a,q)\right)\cdot x ,
\]
because every scalar entropic certainty equivalent fixes constants. Thus, matching the deterministic restriction requires \(\int_\Theta q dm(a,q)=\ell\). The representation theorem says that, apart from matching the deterministic barycenter \(\ell\), an ACF is exactly a positive mixture of the scalar entropic projections in the profile.

\subsection{Representation theorem}

I can now state the representation. I first formulate the result for ACFs, rather than only for normalized certainty equivalents, because the unnormalized version cleanly separates scale from shape. A finite positive measure on \(\Theta\) assigns weight to scalar evaluators indexed by pairs \((a,q)\). The coordinate \(q\) selects a normalized positive valuation of the vector payoff, while \(a\) selects the scalar entropic risk sensitivity applied after that valuation. The total mass of the measure records the scale of the functional; when the functional is a normalized certainty equivalent, this total mass is one.

\begin{restatable}{theorem}{maintheorem}\label{thm:scalar-main}
Let \(\Theta=\overline{\R}\times Q\), and let \(\Phi\colon L^\infty(V)\to\R\) be an ACF, with the deterministic restriction \(\ell(x)=\Phi(x)\). Then there exists a finite positive Borel measure \(m\) on \(\Theta\) such that \(\Phi(X)=\int_\Theta K_a(q\cdot X) dm(a,q)\) for every \(X\in L^\infty(V)\). Moreover, \(\int_\Theta q dm(a,q)=\ell\) so that if \(\ell(e)=1\), then \(m\) is a probability measure.

Conversely, every finite positive Borel measure \(m\) on \(\Theta\) defines an ACF by \(\Phi_m(X)=\int_\Theta K_a(q\cdot X) dm(a,q)\), with the deterministic restriction \(\ell_m=\int_\Theta q dm(a,q)\), so that \(\Phi_m\) is an ACE if and only if \(m\) is a probability measure.\end{restatable}

The first part of the theorem says that the entropic projection profile is exhaustive. An ACE has no additional multidimensional risk component beyond positive mixtures of the scalar evaluations \(X\mapsto K_a(q\cdot X)\). To evaluate \(X\), the representation first views the vector payoff through a normalized positive valuation \(q\), producing the scalar risk \(q\cdot X\). It then evaluates this scalar risk using the entropic certainty equivalent \(K_a\). Finally, it aggregates these scalar evaluations according to the representing measure \(m\). For normalized certainty equivalents, \(m\) is a probability measure, so this aggregation can be understood literally as averaging over as-if latent scalar evaluators. For a general ACF, the same representation holds with an arbitrary positive scale.

The two coordinates of \(\Theta\) have distinct economic roles. The valuation \(q\) determines which aggregate of the vector payoff is being evaluated: prices across goods, welfare weights across groups, date valuations across time, or more generally a normalized shadow valuation compatible with the cone \(C\). The parameter \(a\) determines how risk in that scalar aggregate is treated. Negative \(a\)s are risk-averse, \(a=0\) is risk-neutral, positive \(a\)s are risk-seeking, and the endpoints \(a=\pm\infty\) are worst- and best-case evaluations. Accordingly, the representation decomposes multidimensional risky evaluation into two familiar one-dimensional ingredients: a valuation direction and a risk attitude along that direction.

The barycenter condition is the link back to deterministic choice. If \(x\in V\) is deterministic, then \(q\cdot x\) is a constant, and every scalar entropic certainty equivalent fixes constants:
\(K_a(q\cdot x)=q\cdot x\). Therefore, the risk-sensitivity coordinate \(a\) disappears on deterministic outcomes, and the representing measure affects deterministic choice only through the average of its \(q\)-coordinates.

Deterministic choices identify the shadow valuation \(\ell\), but they do not identify how \(\ell\) is decomposed into valuation directions \(q\) or how much weight is placed on different scalar risk sensitivities \(a\). The representation describes all ways of filling in this missing risky structure consistent with the benchmark requirements.

The converse direction is equally important. It says that the formula is not merely necessary, but is also sufficient. Any finite positive mixture of scalar entropic projections is automatically law-invariant, cone-monotone, and additive across independent sums. If the measure has total mass one, the functional is normalized in the numeraire direction and, hence, is an ACE.

The representing measure in \Cref{thm:scalar-main} need not be unique. The simplest redundancy comes from the mass assigned to \(\left\{0\right\}\times Q\). Since \(K_0(q\cdot X)=q\cdot\E X\), this part of the representation depends only on the barycenter of that mass, not on how the mass is distributed across \(Q\).

\begin{corollary}\label[corollary]{cor:preference-representation}
A preference \(\succeq\in\mathscr E\) admits a representation by a probability measure \(m\in\mathcal P(\Theta)\):\footnote{Here and below, \(\mathcal P(S)\) denotes the Borel probability measures on \(S\).} \(\Phi_\succeq(X)=\int_\Theta K_a(q\cdot X)dm(a,q)\). Its deterministic valuation is \(\int_\Theta qdm(a,q)\). Conversely, every \(m\in\mathcal P(\Theta)\) defines a preference in \(\mathscr E\) by comparing
these integrals.
\end{corollary}

In this light, the theorem identifies a class of mixture representations, not a unique structural distribution of latent types. Risky choices may distinguish among decompositions of the deterministic valuation that are invisible on certain outcomes, but the measure \(m\) is an as-if representing measure unless we impose additional identifying restrictions.

\subsection{Proof Sketch}

I now provide the proof idea behind \Cref{thm:scalar-main}, leaving the auxiliary results and details to \Cref{proofofrepthm}. The only difficult direction of \Cref{thm:scalar-main} is the representation of an arbitrary ACF \(\Phi\). The zero-scale case is immediate: if the deterministic restriction satisfies \(\ell(e)=0\), then \(\ell=0\), and monotonicity together with the fact that every bounded risk is order-bounded between \(-Me\) and \(Me\) forces \(\Phi\equiv 0\). Thus, the interesting case is \(\ell(e)>0\).

The first step is a monotonicity upgrade. ACFs are by definition monotone for vector stochastic dominance, but the representation requires monotonicity with respect to \textit{pointwise} dominance of profiles. Suppose \(\kappa_X\geq \kappa_Y\). Adding \(\varepsilon e\) to \(X\) shifts every scalar projection by \(\varepsilon\), because \(q(e)=1\). This strict shift lets me use the powerful machinery developed by \citet{Fritz2024}; in particular, his catalytic stochastic-order theorem \citep[Theorem 5.6]{Fritz2024}: for every \(\varepsilon>0\), there is an independent bounded background risk \(Z\) such that
\(X+\varepsilon e+Z \succeq_{\mathrm{st}} Y+Z\).

Monotonicity and independent additivity of \(\Phi\) then provide
\[
    \Phi(X)+\varepsilon\ell(e)+\Phi(Z)
    =
    \Phi(X+\varepsilon e+Z)
    \geq
    \Phi(Y+Z)
    =
    \Phi(Y)+\Phi(Z),
\]
and, letting \(\varepsilon\downarrow0\), 
\[
    \kappa_X\geq \kappa_Y
    \quad\Longrightarrow\quad
    \Phi(X)\geq \Phi(Y).
\]
Consequently, every ACF is monotone not only for stochastic dominance of vector risks, but also for pointwise dominance of their entropic projection profiles.

The second step turns \(\Phi\) into a positive linear functional on a space of profile functions. Let \(\mathcal S\) be the additive semigroup generated by profiles:
\[
    \mathcal S
    =
    \left\{
    \kappa_{X_1}+\cdots+\kappa_{X_n}\colon
    n\geq0,\ X_i\in L^\infty(V)
    \right\}.
\]
Because profiles are additive under independent sums, each element of \(\mathcal S\) is itself the profile of an independent sum. Define
\[
    F\left(\sum_{i=1}^n \kappa_{X_i}\right)
    =
    \sum_{i=1}^n \Phi(X_i).
\]
Thanks to the monotonicity upgrade, this is well-defined. Moreover, \(F\) is monotone, and I can deduce a bound that allows me to extend \(F\) from the semigroup of profiles to the closed real linear span of those profiles inside \(C(\Theta)\). Note that this extension is positive, and the span contains the constant functions because \(\kappa_{ce}\) is the constant function \(c\). I then extend this functional to all of \(C(\Theta)\), before concluding--via the Riesz-Markov-Kakutani representation theorem--that there is a finite positive Borel measure \(m\) on \(\Theta\) such that
\[
    \Phi(X)
    =
    \int_\Theta \kappa_X(a,q) dm(a,q)
    =
    \int_\Theta K_a(q\cdot X) dm(a,q)
\]
for every \(X\in L^\infty(V)\).

In turn, I obtain the barycenter condition by evaluating the representation on deterministic vectors. If \(x\in V\), then \(q\cdot x\) is constant and every \(K_a\) fixes constants, so
\[
    \ell(x)
    =
    \Phi(x)
    =
    \int_\Theta q\cdot x dm(a,q)
    =
    \left(\int_\Theta q dm(a,q)\right)\cdot x .
\]
Since this holds for every \(x\), we have \(\int_\Theta q dm(a,q)=\ell\), and by evaluating at \(e\), \(m(\Theta)=\ell(e)\), because \(q(e)=1\) for every \(q\in Q\).

The converse direction of the theorem is straightforward. Given a finite positive Borel measure \(m\) on \(\Theta\), the formula \(\Phi_m(X)=\int_\Theta K_a(q\cdot X) dm(a,q)\) defines a finite-valued functional because profiles are bounded. Law-invariance is inherited from the scalar evaluations. Additivity follows from the additivity of \(K_a\) over independent scalar sums. Cone-monotonicity follows because every \(q\in Q\) is \(C\)-positive, so vector stochastic dominance of \(X\) over \(Y\) implies scalar first-order stochastic dominance of \(q\cdot X\) over \(q\cdot Y\), and each \(K_a\) is monotone for scalar first-order stochastic dominance. Finally, on deterministic vectors,
\(\Phi_m(x)
    =
    \left(\int_\Theta q dm(a,q)\right)\cdot x\), and \(\Phi_m(ce)=cm(\Theta)\). All in all, \(\Phi_m\) is an ACE exactly when \(m(\Theta)=1\).

\section{Which Comparisons Are Robust?}\label{sec:robust}

\Cref{thm:scalar-main} and \Cref{cor:preference-representation} describe one evaluator at a time. I now ask which comparisons survive when the representing measure is varied. The first question holds deterministic tradeoffs fixed. If choices among certain vectors identify \(\ell\in Q\), then the normalized risky evaluations consistent with this deterministic valuation are represented by measures in
\[
  \mathcal M_\ell\coloneqq
  \left\{
  m\in\mathcal P(\Theta)\colon
  \int_\Theta q dm(a,q)=\ell
  \right\}.
\]
Deterministic choices restrict the compatible risky evaluations to those represented by measures in \(\mathcal M_\ell\), but they do not choose among these evaluations.

The distinction matters even in the simplest two-dimensional example. Let \(V=\R^2\), \(C=\R_+^2\), \(e=(1,1)\), and \(Q=\Delta_1\) (the \(1\)-simplex). Fix \(\ell=(1/2,1/2)\), let \(Y=(1/2,1/2)\) be deterministic, and let
\[
  X=
  \begin{cases}
  (1,0),&\text{with probability }1/2,\\
  (0,1),&\text{with probability }1/2.
  \end{cases}
\]
For any \(a\in\overline{\R}\), let \(m=\delta_{(a,\ell)}\). Then \(m\in\mathcal M_\ell\). This measure aggregates first: since \(\ell\cdot X=\ell\cdot Y=1/2\), it ranks the two risks equally, \(\Phi_m(X)=\Phi_m(Y)\).

Now fix finite \(a<0\), and let
\[
  m'=
  \frac12\delta_{\left(a,\left(1,0\right)\right)}
  +
  \frac12\delta_{\left(a,\left(0,1\right)\right)}.
\]
Again \(m'\in\mathcal M_\ell\), so deterministic tradeoffs are unchanged. But this measure evaluates the two coordinate risks separately and then averages. Under each coordinate projection, \(X\) is a nondegenerate Bernoulli random variable with mean \(1/2\), while \(Y\) is constant at \(1/2\). Since \(a<0\), Jensen's inequality implies \(K_a(q\cdot X)<1/2=K_a(q\cdot Y)\) for \(q=(1,0)\) and \(q=(0,1)\) and so \(\Phi_{m'}(X)<\Phi_{m'}(Y)\). Of course, for (finite) \(a>0\), the same construction reverses the strict inequality.

The measures \(m\) and \(m'\) have the same barycenter and, therefore, the same deterministic restriction. Their disagreement concerns whether risk is evaluated after aggregation, or separately along more concentrated valuation directions. This motivates the fixed-\(\ell\) robust order: which comparisons survive every measure in \(\mathcal M_\ell\)?

There is also a natural unrestricted question: which comparisons survive every \(\succeq\in\mathcal E\), allowing deterministic valuations to vary? In that case, each Dirac measure \(\delta_{(a,q)}\) represents the ACE \(W\mapsto K_a(q\cdot W)\), and, as a result, a preference in \(\mathcal E\) with deterministic valuation \(q\). The common order then becomes pointwise dominance of the entropic projection profile.

For \(\ell\in Q\), define
\[
\mathscr E_\ell\coloneqq
\left\{
\succeq\in\mathscr E\colon \Phi_\succeq(x)=\ell(x)\ \text{for every }x\in V
\right\}.
\]
\begin{definition}
For \(X,Y\in L^\infty(V)\), define \(X\succeq_\kappa Y\) if \(\kappa_X(a,q)\ge\kappa_Y(a,q)\) for every \((a,q)\in\Theta\).
\end{definition}

\begin{theorem}\label{thm:profile-order}
For \(X,Y\in L^\infty(V)\), the following are equivalent:
\begin{enumerate}[noitemsep,label=(\roman*)]
  \item\label{profile-orderi} \(X\succeq_\kappa Y\);
  \item\label{profile-orderii} \(\Phi(X)\ge\Phi(Y)\) for every ACF \(\Phi\);
  \item\label{profile-orderiii} \(X\succeq Y\) for every \(\succeq\in\mathscr E\).
\end{enumerate}
For each fixed \(\ell\in Q\),
\[
X\succeq Y\ \text{for every }\succeq\in\mathscr E_\ell
\quad \iff \quad
\int_\Theta\left(\kappa_X-\kappa_Y\right)dm\ge0\ \text{for every }m\in\mathcal M_\ell.
\]
\end{theorem}

The unrestricted order is stronger than the fixed-\(\ell\) order. It asks for unanimity across all compatible deterministic valuations and decompositions of those valuations. Since every scalar coordinate \((a,q)\) can be isolated by a Dirac representing measure, unrestricted unanimity is exactly pointwise dominance of \(\kappa\). In contrast, the Dirac measure \(\delta_{(a,q)}\) belongs to \(\mathcal M_\ell\) if and only if \(q=\ell\).

The unrestricted profile order also has a background-risk interpretation. A comparison is catalytically justified if, after adding an independent multidimensional background risk, ordinary vector stochastic dominance ranks \(X\) above \(Y\). The next theorem says that the unrestricted robust order is the closure of the order generated in this way.

\begin{definition}
Write \(X\succeq_{\cat}Y\) if there exists \(Z \in L^\infty(V)\), independent of the pair \((X,Y)\), such that
\(X+Z\succeq_{\mathrm{st}}Y+Z\). For fixed \(X\), define its \textit{catalytic lower set} by
\[
\mathcal L_{\cat}(X)\coloneqq
\left\{
Y'\in L^\infty(V)\colon X\succeq_{\cat}Y'
\right\}.
\]
\end{definition}

\begin{theorem}\label{thm:unanimous-background}
For \(X,Y\in L^\infty(V)\), \(X\succeq_{\cat}Y\Longrightarrow X\succeq_\kappa Y\).
Conversely,
\[
X\succeq_\kappa Y\quad\Longrightarrow\quad X+\varepsilon e\succeq_{\cat}Y\ \text{for every }\varepsilon>0.
\]
Finally,
\[
X\succeq_\kappa Y
\quad \iff \quad
Y\in\overline{\mathcal L_{\cat}(X)}^{\|\cdot\|_{\infty,e}}.
\]
Consequently,
\[
X\succeq Y\ \text{for every }\succeq \in \mathscr E
\quad \iff \quad
Y\in\overline{\mathcal L_{\cat}(X)}^{\|\cdot\|_{\infty,e}}.
\]
\end{theorem}

The theorem says that the order common to all \(\succeq\in\mathscr E\) is the closed order generated by adding independent multidimensional background risk. Background risk, therefore, plays two roles: invariance to independent background risk engenders the representation in \Cref{thm:scalar-main} and \Cref{cor:preference-representation}, while adding independent background risk generates the unrestricted robust order.

\section{Two Applications}\label{sec:apps}
\subsection{Welfare Weights Under Risk}\label{sec:welfare}

A natural application is distributional social welfare. A risky policy can be viewed as a vector \(X\in L^\infty(\R^d)\), where coordinate \(i\) is the payoff to group \(i\). Pareto improvements are coordinatewise, so the relevant cone is \(C=\R_+^d\). Taking \(e=(1,\ldots,1)\) uses a uniform transfer as the numeraire. In a deterministic welfare analysis, a vector of weights \(\ell\in\Delta_{d-1}\) ranks certain allocations by \(\ell\cdot x\). The question is: once these deterministic welfare weights are fixed, which local comparisons of risky distributional policies survive all representations with those weights and \(a\le0\)?

The normalized dual base \(Q\) is
\[
\Delta_{d-1}\coloneqq\left\{q\in\R_+^d\colon\sum_{i=1}^dq_i=1\right\},
\]
and so \(q\in Q\) is a vector of welfare weights, and \(q\cdot X\) is the weighted utilitarian payoff of the risky allocation \(X\) in each state. The representation below says that a Pareto-monotone, background-risk-invariant planner may evaluate risky allocations by mixing entropic evaluations of such weighted utilitarian payoffs.

\begin{corollary}\label{cor:social-welfare}
Let \(\Phi\colon L^\infty(\R^d)\to\R\) be law-invariant, Pareto-monotone, and additive under independent sums. If its deterministic restriction is \(\Phi(x)=\ell\cdot x\), with \(\ell\in\Delta_{d-1}\), then there is \(m\in\mathcal P(\Theta)\) such that
\[
\Phi(X)=\int_{\Theta}K_a(q\cdot X)dm(a,q)
\qquad\text{and}\qquad
\int_{\Theta}qdm(a,q)=\ell.
\]
Conversely, every such \(m\) defines a law-invariant, Pareto-monotone, additive certainty equivalent by the displayed formula, with the deterministic restriction \(x\mapsto\ell\cdot x\).
\end{corollary}

\begin{proof}
This is \Cref{thm:scalar-main} specialized to \(C=\R_+^d\), for which \(Q=\Delta_{d-1}\). Since \(\ell(e)=1\), the representing measure has total mass one. The converse is the converse direction of \Cref{thm:scalar-main}.
\end{proof}

\Cref{prop:preference-mas-bridge} tells us that these functionals are exactly the certainty-equivalent functionals of Pareto-monotone, background-risk-invariant planners. Such a planner evaluates risky distributional policies as a mixture of entropic evaluations of weighted utilitarian aggregates. The deterministic welfare weights are only the average latent weights \(\ell=\int_{\Theta}qdm(a,q)\). Thus, fixing \(\ell\) fixes the planner's ranking of certain allocations, but not the planner's local ranking of risky reforms around those allocations.

For the local comparisons, I focus on the \(a\le0\) side of the representation. Consequently, I make the standing assumption that the representing measures satisfy \(\int_{\Theta}qdm(a,q)=\ell\), \(m\left(\left\{-\infty,+\infty\right\}\times\Delta_{d-1}\right)=0\), \(\int_{\R\times\Delta_{d-1}}|a|dm(a,q)<\infty\), and \(a \leq 0\) \(m\)-a.s.
I exclude the endpoint cases only so that I may take derivatives at deterministic allocations. For such an \(m\),  I write \(\Phi_m(X)\coloneqq\int_{\Theta}K_a(q\cdot X)dm(a,q)\) and define
\(\Gamma\coloneqq\int_{\R\times\Delta_{d-1}}(-a)qq^\top dm(a,q)\).

If \(U,V\in L^\infty(\R^d)\), \(\mu\in\R^d\), and \(\ell\cdot\E U=\ell\cdot\E V\), then
\[
\left.\frac{d}{dt}\left(\Phi_m(\mu+tU)-\Phi_m(\mu+tV)\right)\right|_{t=0}=0
\]
and, writing \(\Sigma_U=\operatorname{Var}(U)\) and \(\Sigma_V=\operatorname{Var}(V)\),
\[
\left.\frac{d^2}{dt^2}\left(\Phi_m(\mu+tU)-\Phi_m(\mu+tV)\right)\right|_{t=0}
=
\operatorname{tr}\left(\Gamma\left(\Sigma_V-\Sigma_U\right)\right).
\]
Indeed, for finite \(a\), the first two derivatives of \(t\mapsto K_a(q\cdot(\mu+tU))\) at zero are \(q\cdot\E U\) and \(a\operatorname{Var}(q\cdot U)=aq^\top\Sigma_Uq\), with the \(a=0\) case following directly from the expectation.

I conduct a local comparison under fixed welfare weights:
\begin{proposition}\label{prop:fixed-ell-local-welfare}
Let \(U,V\in L^\infty(\R^d)\) satisfy \(\ell\cdot\E U=\ell\cdot\E V\). The following are equivalent.

\begin{enumerate}
\item\label{compareprop1} For every \(\mu\in\R^d\) and every \(m\in\mathcal P(\Theta)\) satisfying the standing restrictions above, \(\left.\frac{d^2}{dt^2}\left(\Phi_m(\mu+tU)-\Phi_m(\mu+tV)\right)\right|_{t=0}\ge0\).

\item\label{compareprop2}
\(\operatorname{Var}(q\cdot U)\le\operatorname{Var}(q\cdot V)\)
for every \(q\in\Delta_{d-1}\) such that \(q_i=0\) whenever \(\ell_i=0\).
\end{enumerate}
\end{proposition}

\begin{proof}
First note that the barycenter condition forces zero-weight groups to receive zero latent weight. If \(\int_{\Theta}qdm(a,q)=\ell\) and \(\ell_i=0\), then
\[
0=\ell_i=\int_{\Theta}q_idm(a,q),
\]
and \(q_i\ge0\) on \(\Delta_{d-1}\). Hence, \(q_i=0\) for \(m\)-a.e. \((a,q)\).

From the derivative taken above,
\[
\left.\frac{d^2}{dt^2}\left(\Phi_m(\mu+tU)-\Phi_m(\mu+tV)\right)\right|_{t=0}
=
\int_{\R\times\Delta_{d-1}}(-a)q^\top\left(\Sigma_V-\Sigma_U\right)qdm(a,q).
\]
Since \(-a\ge0\), this expression is nonnegative for every such \(m\) whenever
\[
q^\top\left(\Sigma_V-\Sigma_U\right)q\ge0
\]
for every \(q\in\Delta_{d-1}\) such that \(q_i=0\) whenever \(\ell_i=0\). This is exactly \ref{compareprop2}, since \(q^\top\Sigma_Uq=\operatorname{Var}(q\cdot U)\) and \(q^\top\Sigma_Vq=\operatorname{Var}(q\cdot V)\).

Conversely, suppose there is \(\hat q\in\Delta_{d-1}\), with \(\hat q_i=0\) whenever \(\ell_i=0\), such that
\[
\hat q^\top\left(\Sigma_V-\Sigma_U\right)\hat q<0.
\]
For \(\varepsilon>0\) small enough that \(\ell_i-\varepsilon\hat q_i\ge0\) for every \(i\), set \(r\coloneqq\frac{\ell-\varepsilon\hat q}{1-\varepsilon}\), so that \(r\in\Delta_{d-1}\), and
\(m_\varepsilon\coloneqq\varepsilon\delta_{(-1,\hat q)}+(1-\varepsilon)\delta_{(0,r)}\) satisfies the restrictions above. But for this measure,
\[
\int_{\R\times\Delta_{d-1}}(-a)q^\top\left(\Sigma_V-\Sigma_U\right)qdm_\varepsilon(a,q)
=
\varepsilon\hat q^\top\left(\Sigma_V-\Sigma_U\right)\hat q<0.
\]
i.e., \ref{compareprop1} fails.
\end{proof}
In short, once deterministic welfare weights are fixed, a local comparison is unanimous across all such mixtures exactly when every compatible welfare-weighted aggregate is less variable under \(U\) than under \(V\).\footnote{Observe also that if \(\ell_i=0\), then group \(i\) is excluded from these aggregates. \textit{Viz.}, a group that receives zero deterministic welfare weight cannot receive positive latent welfare weight only under risk.}

\subsection{Dated Streams}
My representation also applies to dated streams. Interpret the coordinates of \(V\) as deliveries of goods, currencies, or accounts at different dates. The numeraire \(e\) is a benchmark stream, and \(\Phi(X)\) is measured in units of that stream. An element \(q\in Q\) is a normalized positive valuation of the dated coordinates--a present-value functional--so that \(q\cdot X\) is the scalar payoff obtained from the stream using \(q\). The scalar certainty equivalent \(K_a\) evaluates risk in that payoff, and a representing measure \(m\) averages across latent dated valuations and risk sensitivities. As before, deterministic choices identify only the average valuation \(\ell=\int_{\Theta}qdm(a,q)\).

This dated-stream interpretation provides a risky analog of multiple-discount-rate unanimity.\footnote{Cf. \cite{ChambersEchenique2018,ChambersEchenique2020}, who study deterministic one-good streams and disagreement over exponential discount factors. This exercise is also different from the time-lottery one in \cite{MPST2024}. There, the uncertainty is the delivery date of a fixed reward.} In the one-good case with a constant benchmark stream, an exponential discount vector is a special positive valuation of dated coordinates, with coordinates proportional to \(1,\delta,\delta^2,\ldots\). For deterministic streams, requiring agreement for every such valuation asks that \(q\cdot x\geq q\cdot y\) for every exponential \(q\).

For risky streams, the comparison must also survive risk. If we restrict the dated valuations to exponential discount vectors and allow deterministic valuations to vary, the risky analog is the requirement that \(K_a(q\cdot X)\geq K_a(q\cdot Y)\) for every scalar risk sensitivity \(a\) and every exponential discount valuation \(q\). When \(X\) and \(Y\) are deterministic, every \(K_a\) fixes constants, so the test collapses to the previous present-value one. Accordingly, deterministic multiple-discount-rate unanimity is the deterministic-stream case of the risky profile order.

If we allow all \(q\in Q\), rather than only exponential discount valuations, and allow deterministic valuations to vary, the same condition is the full-profile test from \secref{sec:robust}. If the deterministic valuation is fixed at \(\ell\), only mixtures whose average dated valuation is \(\ell\) are permitted.

\appendix

\section{Omitted Proofs}\label[appendix]{app:proofs}

\subsection{Proofs of \texorpdfstring{\Cref{lem:dual-base-order-unit,lem:Q-compact}}{lemmas}}
\begin{proof}[Proof of \Cref{lem:dual-base-order-unit}]
Let \(S=C\cap\{v \in V\colon\|v\|=1\}\). Since \(C\) is closed and has nonempty interior, \(S\) is compact and nonempty. Pointedness implies \(0\notin \conv S\): otherwise, a nontrivial convex combination of points of \(C\) would sum to \(0\), which is impossible in a pointed cone.

By the strict separating hyperplane theorem, applied to the compact convex set \(\conv S\) and the point \(0\), there are \(\beta\in V^*\) and \(\alpha>0\) such that \(\beta(s)\ge\alpha\) for every \(s\in S\). Consequently, \(\beta(c)\ge\alpha\|c\|\) for every \(c\in C\). If \(r\in V^*\) and \(\|r\|_*<\alpha\), then, for every \(c\in C\),
\[
  \left(\beta+r\right)(c)\ge \alpha\|c\|-\|r\|_*\|c\|\ge0.
\]
Thus, \(\beta\in\operatorname{int}C^*\), and so \(C^*\setminus\left\{0\right\}\ne\varnothing\).

If \(q\in C^*\setminus\left\{0\right\}\), then \(q(e)>0\). Indeed, if \(q(e)=0\), then, for all sufficiently small \(h\in V\), both \(e+h\) and \(e-h\) belong to \(C\), so \(0\le q(e+h)=q(h)\) and \(0\le q(e-h)=-q(h)\). Thus, \(q(h)=0\) for all sufficiently small \(h\), and, therefore, \(q=0\), a contradiction. We conclude that \(\beta/\beta(e)\in Q\), so \(Q\) is nonempty. Moreover, as \(q(e) > 0\) (whenever \(q \in C^* \setminus \left\{0\right\}\)), we can normalize every nonzero \(p\in C^*\) by \(p(e)\): \(p/p(e)\in Q\). Accordingly, \(C^*=\left\{\lambda q\colon \lambda\ge0,\ q\in Q\right\}\).

Finally, since \(e\in\operatorname{int}C\), there is \(\delta>0\) such that \(e+h\in C\) whenever \(\|h\|<\delta\). If \(v\in V\), choose \(M>0\) with \(\|v/M\|<\delta\). Then \(e-v/M\in C\) and \(e+v/M\in C\), so \(v\le_C Me\) and \(-Me\le_C v\). The same argument shows that \(\left[-e,e\right]\) contains a neighborhood of \(0\). If \(B\subseteq V\) is bounded, choose \(M>0\) such that \(\|v/M\|<\delta\) for every \(v\in B\). Then \(B\subseteq\left[-Me,Me\right]\).
\end{proof}

\begin{proof}[Proof of \Cref{lem:Q-compact}]
\(Q\) is closed and convex, and it is nonempty by \Cref{lem:dual-base-order-unit}. It remains to prove boundedness. If \(Q\) were unbounded, there would be \(q_n\in Q\) with \(\|q_n\|_*\to\infty\). Passing to a subsequence, \(q_n/\|q_n\|_*\to q\) for some \(q\in V^*\) with \(\|q\|_*=1\). Since \(C^*\) is closed, \(q\in C^*\). But
\[
  q(e)=\lim_n\frac{q_n(e)}{\|q_n\|_*}=0,
\]
contradicting \Cref{lem:dual-base-order-unit}. Thus, \(Q\) is bounded, \textit{ergo} compact.
\end{proof}

\subsection{Proof of \texorpdfstring{\Cref{prop:preference-mas-bridge}}{proposition}}
\begin{proof}[Proof of \Cref{prop:preference-mas-bridge}]
Let \(\succeq\in\mathscr E\). We first show that the numeraire line is ordered by the usual order on \(\R\): \(ce\succeq c'e\) \(\iff\) \(c \geq c'\).

If \(c\ge c'\), then \(ce\succeq_{\mathrm{st}}c'e\), which implies \(ce\succeq c'e\) due to cone monotonicity. Conversely, suppose \(c<c'\) and \(ce\succeq c'e\). Since \(c'e\succeq_{\mathrm{st}}ce\), \(c'e\succeq ce\), again by cone monotonicity. Accordingly, \(ce\sim c'e\). But then the deterministic risk \(ce\) has two certainty equivalents, \(c\) and \(c'\), contradicting uniqueness.

Since \(X\sim \Phi_\succeq(X)e\) and \(Y \sim \Phi_\succeq(Y)e\), 
transitivity delivers
\[
  X\succeq Y
  \quad\iff\quad
  \Phi_\succeq(X)e\succeq \Phi_\succeq(Y)e
  \quad\iff\quad
  \Phi_\succeq(X)\ge\Phi_\succeq(Y),
\]
and so \(\Phi_\succeq\) represents \(\succeq\). In particular, law-invariance and cone monotonicity of \(\succeq\) imply law-invariance and cone monotonicity of \(\Phi_\succeq\).

If \(X,Y\) are independent, then \(X\sim \Phi_\succeq(X)e\). Adding the independent background risk \(Y\), background-risk invariance yields
\[
  X+Y\sim \Phi_\succeq(X)e+Y \sim
  \left(\Phi_\succeq(X)+\Phi_\succeq(Y)\right)e.
\]
By uniqueness of the certainty equivalent,
\[
  \Phi_\succeq(X+Y)=\Phi_\succeq(X)+\Phi_\succeq(Y).
\]
Finally, \(ce\sim ce\), so by uniqueness, \(\Phi_\succeq(ce)=c\) for all \(c\in\R\). We conclude that \(\Phi_\succeq\) is an ACE.

Conversely, let \(\Phi\) be an ACE and define \(\succeq_\Phi\) by
\[
  X\succeq_\Phi Y
  \quad\iff\quad
  \Phi(X)\ge\Phi(Y).
\]
Completeness, reflexivity, and transitivity follow from the usual order on \(\R\). Law-invariance and cone monotonicity follow from the corresponding properties of \(\Phi\). Since \(\Phi\) is normalized, \(\Phi(\Phi(X)e)=\Phi(X)\), so \(X\sim_\Phi \Phi(X)e\). If also \(X\sim_\Phi ce\), then \(\Phi(X)=\Phi(ce)=c\), so the certainty equivalent is unique.

If \(Z\) is independent of the pair \((X,Y)\), by independent additivity, \(\Phi(X+Z)=\Phi(X)+\Phi(Z)\)
and \(\Phi(Y+Z)=\Phi(Y)+\Phi(Z)\). Accordingly,
\[\Phi(X)\ge\Phi(Y)
  \quad\iff\quad
  \Phi(X+Z)\ge\Phi(Y+Z),
\]
which is invariance to independent background risk. We conclude the correspondence.\end{proof}

\subsection{Proofs of \texorpdfstring{\Cref{lem:profile-lipschitz,lem:profile-additivity}}{lemmas}}
\begin{proof}[Proof of \Cref{lem:profile-lipschitz}]
I prove this lemma through a sequence of claims that establish the result for finitely-valued risks, which I then use in an approximation argument. 
\begin{claim}
    \(\kappa_X\) is bounded.
\end{claim}
\begin{proof} Fix \(X\in L^\infty(V)\). Since bounded sets are order-bounded by \Cref{lem:dual-base-order-unit}, choose \(M>0\) such that \(-Me\le_C X\le_C Me\) a.s. Then, for every \(q\in Q\), \(-M\le q\cdot X\le M\) a.s., because \(q\in C^*\) and \(q(e)=1\). By definition, each \(K_a\) is monotone and satisfies \(K_a(c)=c\) for a constant \(c\). Hence, \(-M\le K_a(q\cdot X)\le M\) for every \((a,q)\in\Theta\), i.e., \(\kappa_X\) is bounded.
\end{proof}

\begin{claim}\label[claim]{claim:lip}
    \(\left\|\kappa_X-\kappa_Y\right\|_\infty\le\left\|X-Y\right\|_{\infty,e}\).
\end{claim}
\begin{proof}
    First, note that if \(Z\) and \(W\) are bounded real random variables and \(Z\le W+r\) a.s. for some \(r\ge0\), then \(K_a(Z)\le K_a(W)+r\) for every \(a\in\overline{\R}\). For \(a>0\), this follows from \(\E\mathrm e^{aZ}\le\mathrm e^{ar}\E\mathrm e^{aW}\). For \(a<0\), the inequality \(Z\le W+r\) yields \(\E\mathrm e^{aZ}\ge\mathrm e^{ar}\E\mathrm e^{aW}\), and division by \(a<0\) reverses the inequality after taking logs. The cases \(a=0,+\infty,-\infty\) follow from the definitions.\footnote{For \(a=0\), this is \(\E Z\le\E W+r\); for \(a=+\infty\), this is \( \esssup Z\le \esssup \left(W+r\right)= \esssup W+r\); and for \(a=-\infty\), this is \( \essinf Z\le \essinf \left(W+r\right)= \essinf W+r\).} Applying the same argument with \(Z\) and \(W\) reversed delivers
\[\tag{\(1\)}\label{in1}
\left|K_a(Z)-K_a(W)\right|\le\left\|Z-W\right\|_\infty.
\]

Now fix \(X,Y\in L^\infty(V)\), and take \(r>\left\|X-Y\right\|_{\infty,e}\). Then \(-re\le_C X-Y\le_C re\) a.s. Hence, for every \(q\in Q\), \(\left|q\cdot X-q\cdot Y\right|\le r\) a.s. Applying \eqref{in1} with \(Z=q\cdot X\) and \(W=q\cdot Y\), we get \(\left|K_a(q\cdot X)-K_a(q\cdot Y)\right|\le r\) for every \((a,q)\in\Theta\). Taking the supremum over \(\Theta\) and then letting \(r\downarrow\left\|X-Y\right\|_{\infty,e}\) yields
\(\left\|\kappa_X-\kappa_Y\right\|_\infty\le\left\|X-Y\right\|_{\infty,e}\).\end{proof}

\begin{claim}\label[claim]{claim:fincont}
    If \(X\) is finite-valued, then \(\kappa_X \in C(\Theta)\).
\end{claim}
\begin{proof}
    Suppose \(X\) takes the finitely many values \(x_1,\ldots,x_N\in V\), with probabilities \(p_1,\ldots,p_N>0\). For finite \(a\ne0\),
\[
\kappa_X(a,q)=\frac{1}{a}\log\left(\sum_{i=1}^Np_i\mathrm e^{a q\cdot x_i}\right).
\]
For such \(a\), we use the fundamental theorem of calculus to obtain the alternative formula
\[
\kappa_X(a,q)=\int_0^1
\frac{\sum_{i=1}^Np_i\left(q\cdot x_i\right)\mathrm e^{ta q\cdot x_i}}
{\sum_{i=1}^Np_i\mathrm e^{ta q\cdot x_i}}dt.
\]
The integrand is continuous in \((t,a,q)\), and the denominator is always strictly positive. Thus, the right-hand side is continuous in \((a,q)\in\R\times Q\). At \(a=0\), it equals \(\sum_i p_iq\cdot x_i=\E\left[q\cdot X\right]=K_0(q\cdot X)\). Accordingly, \(\kappa_X\) is continuous on \(\R\times Q\).

It remains to check the endpoints. Since \(X\) is finite-valued, \(K_{+\infty}(q\cdot X)=\max_i q\cdot x_i\) and \(K_{-\infty}(q\cdot X)=\min_i q\cdot x_i\). For \(a>0\),
\[
\frac{\log\left(\min_i p_i\right)}{a}
\le
\kappa_X(a,q)-\max_i q\cdot x_i
\le0.
\]
Indeed, after factoring out \(\max_i q\cdot x_i\), the remaining weighted exponential sum lies between \(\min_i p_i\) and \(1\). The bound is uniform in \(q\), and \(q\mapsto\max_i q\cdot x_i\) is continuous. Hence, \(\kappa_X(a,q)\to K_{+\infty}(q\cdot X)\) uniformly in \(q\) as \(a\to+\infty\).

Similarly, for \(a<0\),
\[
0
\le
\kappa_X(a,q)-\min_i q\cdot x_i
\le
\frac{-\log\left(\min_i p_i\right)}{\left|a\right|}.
\]
The bound is uniform in \(q\), and \(q\mapsto\min_i q\cdot x_i\) is continuous. Hence, \(\kappa_X(a,q)\to K_{-\infty}(q\cdot X)\) uniformly in \(q\) as \(a\to-\infty\). Therefore, \(\kappa_X\in C(\Theta)\) whenever \(X\) is finite-valued.
\end{proof}
We finish the proof of the lemma by approximating arbitrary risks by finite-valued risks. Take arbitrary \(X\in L^\infty(V)\). After changing \(X\) on a null set, assume its range is contained in \(\left\{v\in V\colon \left\|v\right\|_e\le M\right\}\) for some \(M<\infty\). This set is compact. For each \(n\), cover it by finitely many open \(\left\|\cdot\right\|_e\)-balls of radius \(1/(2n)\), refine the cover into a finite Borel partition, and choose one point from each nonempty cell. Let \(X_n\) be the chosen point in the cell containing \(X\). Then each \(X_n\) is finite-valued and
\(\left\|X_n-X\right\|_{\infty,e}\le\frac{1}{n}\).

By \Cref{claim:fincont}, \(\kappa_{X_n}\in C(\Theta)\); and by \Cref{claim:lip},
\[
\left\|\kappa_{X_n}-\kappa_X\right\|_\infty\le\left\|X_n-X\right\|_{\infty,e}\le\frac{1}{n}.
\]
Thus, \(\kappa_X\) is the uniform limit of continuous functions on \(\Theta\) and so \(\kappa_X\in C(\Theta)\).
\end{proof}

\begin{proof}[Proof of \Cref{lem:profile-additivity}]
For finite \(a\ne0\), we have \(\E\mathrm e^{a q\cdot\left(X+Y\right)}
=
\E\mathrm e^{a q\cdot X}\E\mathrm e^{a q\cdot Y}\) due to independence.\footnote{Independence is preserved under measurable transformations, and \(q\) is fixed.} Taking logs and dividing by \(a\) yields \(K_a(q\cdot\left(X+Y\right))=K_a(q\cdot X)+K_a(q\cdot Y)\). For \(a=0\), the linearity of \(q\) and the expectation gives us
\[
K_0(q\cdot\left(X+Y\right))
=
\E\left[q\cdot\left(X+Y\right)\right]
=
\E\left[q\cdot X\right]+\E\left[q\cdot Y\right]
=
K_0(q\cdot X)+K_0(q\cdot Y).
\]

For the endpoint \(a=+\infty\), set \(Z\coloneqq q\cdot X\) and \(W\coloneqq q\cdot Y\). The inequality \( \esssup \left(Z+W\right)\le \esssup Z+ \esssup W\) is immediate. For the reverse inequality, fix \(\varepsilon>0\). The events \(\left\{Z> \esssup Z-\varepsilon\right\}\) and \(\left\{W> \esssup W-\varepsilon\right\}\) have positive probability (by the definition of the essential supremum). By independence, their intersection has positive probability. Hence, \( \esssup \left(Z+W\right)\ge \esssup Z+ \esssup W-2\varepsilon\). Letting \(\varepsilon\downarrow0\), we have \[\tag{\(2\)}\label{esssupid} \esssup \left(Z+W\right)= \esssup Z+ \esssup W.\] The essential-infimum identity, \[ \essinf \left(Z+W\right)= \essinf Z+ \essinf W,\] follows by applying the essential-supremum identity \eqref{esssupid} to \(-Z\) and \(-W\). We conclude that the endpoint cases also satisfy \(\kappa_{X+Y}=\kappa_X+\kappa_Y\).

The deterministic identity is immediate because \(q\cdot x\) is constant, and \(K_a\) fixes constants.
\end{proof}

\subsection{Proof of \texorpdfstring{\Cref{lem:detlin}}{lemma}}
\begin{proof}[Proof of \Cref{lem:detlin}]
For deterministic \(x,y\),
\[
  \ell(x+y)=\Phi(x+y)=\Phi(x)+\Phi(y)=\ell(x)+\ell(y).
\]
Also, \(\ell(0)=\ell(0)+\ell(0)\), so \(\ell(0)=0\). If \(c\in C\), then \(c\succeq_{\mathrm{st}}0\), so \(\ell(c)\ge \ell(0)=0\). Since \(e\in\operatorname{int}C\), for all sufficiently small \(h\in V\), both \(e+h\) and \(e-h\) belong to \(C\). Therefore,
\(-\ell(e)\le \ell(h)\le \ell(e)\) on a neighborhood of \(0\). An additive function bounded on a neighborhood is continuous, hence, linear. Positivity on \(C\) yields \(\ell\in C^*\).
\end{proof}

\subsection{Proof of \texorpdfstring{\Cref{thm:scalar-main}}{theorem}}\label[appendix]{proofofrepthm}

I begin with a sequence of auxiliary results. First, I state \citet[Theorem 5.6]{Fritz2024}, before verifying that its conditions are satisfied. In particular, I use exactly the strict-converse part of \cite[Theorem 5.6]{Fritz2024}, specialized to \(G=V=\R^d\), \(G_+=C\), and \(u=e\). In \citeauthor{Fritz2024}'s notation, \(X\le^d Y\) denotes the stochastic preorder between the laws, i.e., \(X\preceq_{st} Y\) in my notation. The strict inequalities in his condition (iii) imply his condition (i), namely the existence of a compactly supported catalyst \(Z\) with \(X+Z\le^d Y+Z\). I record this specialization explicitly.
\begin{theorem}[Theorem 5.6 in \citet{Fritz2024}]\label{thm:fritz-finite}
Let \(A,B\) be compactly supported \(V\)-valued Radon random variables. Suppose that for every nonzero \(t\in C^*\),
\(\E \mathrm e^{t\cdot B}<\E \mathrm e^{t\cdot A}\), \(\E \mathrm e^{-t\cdot B}>\E \mathrm e^{-t\cdot A}\), \(\E[t\cdot B]<\E[t\cdot A]\), \(\max_{x\in\supp B}t\cdot x<\max_{x\in\supp A}t\cdot x\), and \(\min_{x\in\supp B}t\cdot x<\min_{x\in\supp A}t\cdot x\). Then there exists a compactly supported Radon random vector \(Z\), independent of the pair \((A,B)\), such that
\(B+Z\preceq_{\mathrm{st}}A+Z\).
\end{theorem}

\begin{proof}[Proof of \Cref{thm:fritz-finite}]
By \Cref{lem:dual-base-order-unit}, the standing assumptions make \(V\) a finite-dimensional ordered topological abelian group with closed positive cone \(C\) and order unit \(e\), and compact support is automatic for bounded random variables. This is \cite[Theorem 5.6]{Fritz2024}, applied with Fritz's lower variable \(X=B\), upper variable \(Y=A\), topological abelian group \(G=V\), positive cone \(G_+=C\), and order unit \(u=e\). Fritz's condition (iii), with strict inequalities, is precisely the displayed list above: the first two strict inequalities are the positive and negative exponential tests, the third is the mean test, and the last two strict inequalities are the upper- and lower-support tests. Its strict-converse conclusion is condition (i), namely the existence of a compactly supported Radon catalyst \(Z\), independent of the pair \((A,B)\), such that \(B+Z\le^d A+Z\). \end{proof}

Next, I use \Cref{thm:fritz-finite} to convert profile dominance into stochastic dominance after adding an independent catalyst. The theorem requires strict inequalities in its exponential, mean, and support tests. The small shift by \(\varepsilon e\) supplies exactly this strictness: since \(q(e)=1\), adding
\(\varepsilon e\) increases every scalar projection \(q\cdot X\) by
\(\varepsilon\). \Cref{lem:dual-base-order-unit} then lets every nonzero \(t\in C^*\) be written as \(t=\lambda q\), with \(\lambda>0\) and \(q\in Q\), so strict profile dominance
implies all of Fritz's tests. This yields the catalytic comparison I use below.

\begin{lemma}\label[lemma]{lem:catalytic-bridge}
Let \(X,Y\in L^\infty(V)\). If \(\kappa_X(a,q)\ge \kappa_Y(a,q)\) for all \((a,q)\in\Theta\), then for every \(\varepsilon>0\) there exists a bounded random vector \(Z\), independent of the pair \((X,Y)\), such that \(X+\varepsilon e+Z\succeq_{\mathrm{st}}Y+Z\).
\end{lemma}

\begin{proof}[Proof of \Cref{lem:catalytic-bridge}]
Fix \(\varepsilon>0\), and set \(A\coloneqq X+\varepsilon e\) and \(B\coloneqq Y\). Since \(K_a\left(Z+c\right)=K_a\left(Z\right)+c\) and \(q(e)=1\) for every \(q\in Q\),
\[
  \kappa_A(a,q)=K_a\left(q\cdot X+\varepsilon\right)=\kappa_X(a,q)+\varepsilon.
\]
Therefore, the profile dominance we assumed implies strict profile dominance:
\[
  \kappa_A(a,q)>\kappa_B(a,q)
  \qquad
  \forall(a,q)\in\Theta.
\]
The variables \(A\) and \(B\) are compactly supported because \(X\) and \(Y\) are bounded. By \Cref{lem:dual-base-order-unit}, \(e\) is an order unit in Fritz's sense.

Let \(t\in C^*\setminus\left\{0\right\}\). Again by \Cref{lem:dual-base-order-unit}, set \(\lambda\coloneqq t(e)>0\) and \(q\coloneqq t/t(e)\in Q\), so that \(t=\lambda q\). With this in hand, we pursue the five tests from \Cref{thm:fritz-finite}: first, set \(a = \lambda > 0\). Then, \(K_\lambda(q\cdot A)>K_\lambda(q\cdot B)\), i.e.,
\[
\frac{1}{\lambda}\log\E\mathrm e^{\lambda q\cdot A}
>
\frac{1}{\lambda}\log\E\mathrm e^{\lambda q\cdot B}.
\]
Since \(\lambda>0\) and \(t=\lambda q\), this is \(\E\mathrm e^{t\cdot B}<\E\mathrm e^{t\cdot A}\). The symmetric argument using \(a=-\lambda<0\) produces \(\E \mathrm e^{-t\cdot B}>\E \mathrm e^{-t\cdot A}\).

Now evaluate strict profile dominance at \(a=0\) to get \(\E\left[q\cdot A\right]>\E\left[q\cdot B\right]\). Since \(t=\lambda q\) and \(\lambda>0\), this produces \(\E\left[t\cdot B\right]<\E\left[t\cdot A\right]\).

Analogously, at \(a=+\infty\), \(\esssup(q\cdot A)>\esssup(q\cdot B)\). Since the laws are compactly supported, these essential suprema are maxima over supports. Multiplying by \(\lambda>0\),
\(\max_{x\in\supp B}t\cdot x<\max_{x\in\supp A}t\cdot x\). We easily do likewise at \(a=-\infty\).

In sum, all strict inequalities in \Cref{thm:fritz-finite} hold for the lower variable \(B\) and the upper variable \(A\). Hence, there exists a compactly supported \(Z\), independent of the pair \(\left(A,B\right)\), such that \(\Law(B+Z)\le_{\mathrm{st}}\Law(A+Z)\), equivalently \(A+Z\succeq_{\mathrm{st}}B+Z\). Since \(A=X+\varepsilon e\) and \(B=Y\), the same \(Z\) is independent of the pair \(\left(X,Y\right)\), and so \(X+\varepsilon e+Z\succeq_{\mathrm{st}}Y+Z\), delivering the result.\end{proof}

The remaining step is functional analytic. After the catalytic argument establishes monotonicity with respect to the profile order, the proof constructs a positive linear functional on the span of the profile functions. To represent that functional by a measure on \(\Theta\), it must be extended to all of \(C(\Theta)\) without losing positivity. \Cref{lem:positive-extension} provides exactly this extension, which allows us to then apply the Riesz-Markov-Kakutani theorem to the resulting positive linear functional on \(C(\Theta)\).

\begin{lemma}\label[lemma]{lem:positive-extension}
Let \(K\) be compact, and let \(E\subseteq C(K)\) be a linear subspace containing the constants. Then any positive linear functional \(I\colon E\to\R\) extends to a positive linear functional on \(C(K)\).
\end{lemma}

\begin{proof}[Proof of \Cref{lem:positive-extension}]
Since \(I\) is positive and constants belong to \(E\), for every \(u\in E\),
\[
  -\|u\|_\infty\one\le u\le \|u\|_\infty\one \qquad \Longrightarrow \qquad |I(u)|\le I(\one)\|u\|_\infty.
\]
Thus, \(I\) is bounded and \(\|I\|=I(\one)\). By the Hahn-Banach theorem in its norm-preserving linear extension form, extend \(I\) to a bounded linear functional \(L\colon C(K)\to\R\) with the same norm. Then \(L(\one)=I(\one)=\|L\|\). If \(0\le f\le\one\), then \(\|\one-f\|_\infty\le1\), so that \(L(\one-f)\le \|L\|=L(\one)\). Consequently, \(L(f)\ge0\). Scaling delivers \(L(f)\ge0\) for every \(f\ge0\), so \(L\) is positive.
\end{proof}

With these in hand, I prove the main representation theorem:

\maintheorem*

\begin{proof}[Proof of \Cref{thm:scalar-main}]If \(\ell(e)=0\), then \(\ell=0\), since \(\ell\in C^*\) and every nonzero element of \(C^*\) is strictly positive on \(e\). Since every bounded \(X\) is order-bounded between \(-Me\) and \(Me\) for some \(M\), monotonicity delivers
\[
  \Phi(-Me)\le \Phi(X)\le \Phi(Me),
\]
and both endpoints are \(0\). Hence, \(\Phi\equiv0\), and the result holds with \(m=0\). Assume \(\ell(e)>0\).

First we prove a monotonicity upgrade:
\begin{claim}\label[claim]{claim:monoup}
    If \(\kappa_X\ge \kappa_Y\), then \(\Phi(X)\ge\Phi(Y)\).
\end{claim}
\begin{proof}
    Suppose \(\kappa_X\ge \kappa_Y\), in which case by \Cref{lem:catalytic-bridge}, for every \(\varepsilon>0\), there exists independent-of-\((X,Y)\) \(Z\) such that \[
  X+\varepsilon e+Z\succeq_{\mathrm{st}} Y+Z.
\]
By monotonicity and additivity,
\[
  \Phi(X+\varepsilon e)+\Phi(Z)
  =
  \Phi(X+\varepsilon e+Z)
  \ge
  \Phi(Y+Z)
  =
  \Phi(Y)+\Phi(Z),
\]
so that \(\Phi(X)+\varepsilon\ell(e)\ge \Phi(Y)\). Then let \(\varepsilon\downarrow0\) to obtain the claim.\end{proof}

Let \(\mathcal S\subset C(\Theta)\) be the additive semigroup generated by profiles:
\[
  \mathcal S=
  \left\{
  \kappa_{X_1}+\cdots+\kappa_{X_n}:
  n\ge0,\ X_i\in L^\infty(V)
  \right\}.
\]
The empty sum is \(0\). Since profiles add under independent sums, every element of \(\mathcal S\) is itself the profile of an independent sum.

Define \(F \colon \mathcal S\to\R\) by
\[
  F\left(\sum_{i=1}^n\kappa_{X_i}\right)
  =
  \sum_{i=1}^n\Phi(X_i).
\]
\begin{claim}
    The map \(F\) is well-defined.
\end{claim}
\begin{proof}
 If \(\sum_i\kappa_{X_i}=\sum_j\kappa_{Y_j}\), let \(\widetilde X=\sum_i X_i'\) and \(\widetilde Y=\sum_j Y_j'\), where the \(X_i'\)'s and \(Y_j'\)'s are independent copies of the corresponding variables, all jointly independent. Then,
\[
  \kappa_{\widetilde X}=\sum_i\kappa_{X_i}=\sum_j\kappa_{Y_j}=\kappa_{\widetilde Y}.
\]
Applying \Cref{claim:monoup} in both directions yields \(\Phi(\widetilde X)=\Phi(\widetilde Y)\). By law-invariance and additivity, this is exactly \(\sum_i\Phi(X_i)=\sum_j\Phi(Y_j)\).
\end{proof}
Moreover, the map \(F\) is additive on \(\mathcal S\) and is monotone: if \(g,h\in\mathcal S\) and \(g\ge h\) pointwise, then, writing \(g=\kappa_{\widetilde X}\) and \(h=\kappa_{\widetilde Y}\),
\(F(g)=\Phi(\widetilde X)\ge\Phi(\widetilde Y)=F(h)\).

We next prove a Lipschitz estimate:
\begin{claim}\label[claim]{claim:lipschitz}
    For \(g,h \in \mathcal S\), \(|F(g)-F(h)|\le \ell(e)\|g-h\|_\infty\).
\end{claim}
\begin{proof}
    Let \(g,h\in\mathcal S\). Setting \(r=\|g-h\|_\infty\), \(g\le h+r\one\). Since \(r\one=\kappa_{re}\), we have \(h+r\one\in\mathcal S\), and so \(F(g)\le F(h+r\one)=F(h)+r\ell(e)\). Symmetrically, \(F(h)\le F(g)+r\ell(e)\). Combining, we obtain the claim.
\end{proof}

Let \(E_{\Q}\) be the rational linear span of \(\mathcal S\). Equivalently, every \(u\in E_{\Q}\) can be written as \(u=\frac1N(g-h)\) for some integer \(N\ge1\) and some \(g,h\in\mathcal S\). Define
\[
  I(u) \coloneqq \frac1N(F(g)-F(h)).
\]
This is well-defined,\footnote{\(\frac1N(g-h)=\frac1M(g'-h')\) implies \(M g+N h'=M h+N g'\), where \(M g\) denotes the sum of \(M\) copies of \(g\) in \(\mathcal S\); and additivity of \(F\) delivers equality of the corresponding values.} and so \(I\) is a \(\Q\)-linear functional on \(E_{\Q}\).

Furthermore, the functional \(I\) is positive on \(E_{\Q}\). Indeed, if \(u=\frac1N(g-h)\ge0\), then \(g\ge h\), and by the monotonicity of \(F\) on \(\mathcal S\), \(F(g)\ge F(h)\). Thus, \(I(u)\ge0\). The Lipschitz estimate of \Cref{claim:lipschitz} extends to \(E_{\Q}\):
\(|I(u)|\le \ell(e)\|u\|_\infty\).\footnote{Indeed, if \(u=\frac1N(g-h)\), then
\(|I(u)|=\frac1N|F(g)-F(h)|\le\frac{\ell(e)}{N}\|g-h\|_\infty=\ell(e)\|u\|_\infty\).}

Now let \(E\) be the \(\sup\)-norm closure of \(E_{\Q}\),  \(E \coloneqq \overline{E_{\Q}}^{\|\cdot\|_\infty}\subset C(\Theta)\). The space \(E\) is a real linear subspace: if \(u_n,v_n\in E_{\Q}\) converge to \(u,v\in E\), then \(u_n+v_n\to u+v\), so \(u+v\in E\); if \(u_n\in E_{\Q}\) converges to \(u\in E\) and \(\alpha\in\R\), choose rationals \(\alpha_n\to\alpha\), then \(\alpha_nu_n\to\alpha u\), so \(\alpha u\in E\). Our Lipschitz estimate (\Cref{claim:lipschitz}), therefore, lets \(I\) extend uniquely to a continuous real-linear functional on \(E\). The extension is positive.\footnote{Indeed, if \(u\in E\), \(u\ge0\), choose \(u_n\in E_{\Q}\) with \(u_n\to u\). Choose rational \(r_n>0\) such that \(r_n>\|u_n-u\|_\infty\) and \(r_n\to0\). Then \(u_n+r_n\one\ge0\) and so \(I(u_n)+r_n\ell(e)=I(u_n+r_n\one)\ge0\). Letting \(n\to\infty\), we get \(I(u)\ge0\).}

The space \(E\) also contains all constant functions, because \(c\one=\kappa_{ce}\) for every \(c\in\R\). By \Cref{lem:positive-extension}, \(I\) extends to a positive linear functional \(\overline I\colon C(\Theta)\to\R\). By the Riesz-Markov-Kakutani representation theorem, there exists a finite positive Borel measure \(m\) on \(\Theta\) such that \(\overline I(f)=\int_\Theta f dm\) for all \(f\in C(\Theta)\). In particular,
\[
  \Phi(X)=I(\kappa_X)=\overline I(\kappa_X)
  =
  \int_\Theta K_a(q\cdot X) dm(a,q).
\]

Finally, for deterministic \(x\in V\),
\[
  \ell(x)=\Phi(x)
  =
  \int_\Theta q\cdot x dm(a,q)
  =
  \left(\int_\Theta q dm(a,q)\right)\cdot x.
\]
Since this holds for every \(x\), we obtain
\(\int_\Theta q dm(a,q)=\ell\), and by evaluating at \(e\), we obtain
\[
  m(\Theta)=\int_\Theta q(e) dm(a,q)=\ell(e).
\]

Conversely, let \(m\) be a finite positive Borel measure on \(\Theta\), and define
\(\Phi_m(X) \coloneqq \int_\Theta K_a(q\cdot X) dm(a,q)\). Observe that this integral is finite for every bounded \(X\) by
\Cref{lem:profile-lipschitz}. Law-invariance is immediate. If \(X,Y\) are independent, then \(q\cdot X\) and \(q\cdot Y\) are independent for every \(q\in Q\), so \Cref{lem:profile-additivity} guarantees \(\Phi_m(X+Y)=\Phi_m(X)+\Phi_m(Y)\).

If \(X\succeq_{\mathrm{st}}Y\), then \(q\cdot X\) FOSD \(q\cdot Y\) for every \(q\in Q\): for each \(r\in\R\), the set \(\{v \colon q\cdot v\ge r\}\) is closed and \(C\)-upper. Each \(K_a\) is monotone for FOSD, including \(a=0,\pm\infty\), hence, \(K_a(q\cdot X)\ge K_a(q\cdot Y)\) for every \((a,q)\in\Theta\). Integrate against \(m\ge0\) to get \(\Phi_m(X)\ge\Phi_m(Y)\). Finally, for deterministic \(x\), \Cref{lem:profile-additivity} implies
\[
  \ell_m(x)=\Phi_m(x)=\int_\Theta q\cdot x dm(a,q)=\left(\int_\Theta q dm(a,q)\right)\cdot x.
\]
Thus, \(\ell_m=\int_\Theta q dm(a,q)\), and \(\ell_m(e)=m(\Theta)\). Since \(\Phi_m(ce)=c m(\Theta)\) for every \(c\in\R\), \(\Phi_m\) is an ACE if and only if \(m(\Theta)=1\), i.e., if and only if \(m\) is a probability measure.
\end{proof}

I finish by proving \Cref{cor:preference-representation}:
\begin{proof}[Proof of \Cref{cor:preference-representation}]
By \Cref{prop:preference-mas-bridge}, \(\Phi_\succeq\) is an ACE. \Cref{thm:scalar-main} produces a finite positive representing measure \(m\) and the barycenter identity. Since \(\Phi_\succeq(e)=1\), we have \(\ell(e)=1\). Therefore, \(m(\Theta)=\ell(e)=1\), so \(m\in\mathcal P(\Theta)\), and \(\ell\in Q\).

Conversely, if \(m\in\mathcal P(\Theta)\), then \(\Phi_m(X)=\int_\Theta K_a(q\cdot X) dm(a,q)\) is an ACE, since each map \(X\mapsto K_a(q\cdot X)\) is law-invariant, cone-monotone, additive under independent sums, and satisfies \(K_a(q\cdot ce)=c\). Its deterministic restriction is
\(\ell_m=\int_\Theta q dm(a,q)\). \Cref{prop:preference-mas-bridge} then provides an admissible preference, with fixed deterministic valuation \(\ell\) exactly when \(\int_\Theta q dm(a,q)=\ell\).
\end{proof}


\subsection{Proofs of \texorpdfstring{\Cref{thm:profile-order} and \Cref{thm:unanimous-background}}{theorems}}

\begin{proof}[Proof of \Cref{thm:profile-order}]
Suppose first that \(X\succeq_\kappa Y\), and let \(\Phi\) be an ACF. By \Cref{thm:scalar-main}, there is a finite positive Borel measure \(m\) on \(\Theta\) such that
\[
\Phi(W)=\int_\Theta K_a(q\cdot W)dm(a,q)
\]
for every bounded \(W\). Integrating the inequalities \(\kappa_X(a,q)\ge\kappa_Y(a,q)\) yields \(\Phi(X)\ge\Phi(Y)\).

Conversely, suppose \(\Phi(X)\ge\Phi(Y)\) for every ACF \(\Phi\), and fix \((a,q)\in\Theta\). By the converse direction of \Cref{thm:scalar-main}, the Dirac measure \(\delta_{(a,q)}\) defines the ACF
\[
W\mapsto K_a(q\cdot W).
\]
Therefore, \(K_a(q\cdot X)\ge K_a(q\cdot Y)\). Since \((a,q)\) was arbitrary, \(X\succeq_\kappa Y\). This proves \ref{profile-orderi} \(\iff\) \ref{profile-orderii}.

Now suppose \(X\succeq_\kappa Y\), and let \(\succeq \in \mathscr E\). By \Cref{prop:preference-mas-bridge}, \(\Phi_{\succeq}\) is an ACE, hence an ACF. The equivalence of \ref{profile-orderi} and \ref{profile-orderii} delivers \(\Phi_{\succeq}(X)\ge\Phi_{\succeq}(Y)\). Applying \Cref{prop:preference-mas-bridge} again: \(X\succeq Y\).

Conversely, suppose \(X\succeq Y\) for every \(\succeq \in \mathscr E\), and fix \((a,q)\in\Theta\). Set \(m=\delta_{(a,q)}\). By \Cref{cor:preference-representation}, \(m\) defines an admissible preference \(\succeq_m\in\mathscr E\) with certainty equivalent \(\Phi_m(W)=K_a(q\cdot W)\). The assumed unanimity implies \(X\succeq_mY\). By \Cref{prop:preference-mas-bridge}, \(\Phi_m(X)\ge\Phi_m(Y)\), so \(K_a(q\cdot X)\ge K_a(q\cdot Y)\). Since \((a,q)\) was arbitrary, \(X\succeq_\kappa Y\). This proves \ref{profile-orderi} \(\iff\) \ref{profile-orderiii}.

Now fix \(\ell\in Q\). Suppose
\[
\int_\Theta\left(\kappa_X-\kappa_Y\right)dm\ge0\ \text{for every }m\in\mathcal M_\ell,
\]
and let \(\succeq \in \mathscr E_\ell\). By \Cref{cor:preference-representation}, there is \(m\in\mathcal M_\ell\) such that
\[
\Phi_{\succeq}(W)=\int_\Theta K_a(q\cdot W)dm(a,q)
\]
for every bounded \(W\). Hence,
\[
\Phi_{\succeq}(X)-\Phi_{\succeq}(Y)=\int_\Theta\left(\kappa_X-\kappa_Y\right)dm\ge0.
\]
By \Cref{prop:preference-mas-bridge}, \(X\succeq Y\).

Conversely, suppose \(X\succeq Y\) for every \(\succeq \in \mathscr E_\ell\), and fix \(m\in\mathcal M_\ell\). By \Cref{cor:preference-representation}, \(m\) defines an admissible preference \(\succeq_m\in\mathscr E_\ell\) with certainty equivalent
\[
\Phi_m(W)=\int_\Theta K_a(q\cdot W)dm(a,q).
\]
By the assumed unanimity, \(X\succeq_mY\); and by \Cref{prop:preference-mas-bridge}, \(\Phi_m(X)\ge\Phi_m(Y)\). So,
\[
\int_\Theta\left(\kappa_X-\kappa_Y\right)dm
=\Phi_m(X)-\Phi_m(Y)
\ge0.
\]
which proves the fixed-\(\ell\) equivalence.
\end{proof}

\begin{proof}[Proof of \Cref{thm:unanimous-background}]
Suppose \(X\succeq_{\cat}Y\). Choose \(Z \in L^\infty(V)\), independent of \((X,Y)\), such that \(X+Z\succeq_{\mathrm{st}}Y+Z\). Fix \((a,q)\in\Theta\). Since \(q\in Q\subset C^*\), the projection \(v\mapsto q\cdot v\) is \(C\)-increasing, so \(q\cdot X+q\cdot Z\) FOSD \(q\cdot Y+q\cdot Z\). By the definition of \(K_a\), \(K_a\) is monotone for FOSD. Hence,
\(K_a(q\cdot X+q\cdot Z)\ge K_a(q\cdot Y+q\cdot Z)\). 

Since \(Z\) is independent of \((X,Y)\), we have from \Cref{lem:profile-additivity} that
\[
K_a(q\cdot X)+K_a(q\cdot Z)=K_a(q\cdot X+q\cdot Z)\ge K_a(q\cdot Y+q\cdot Z)=K_a(q\cdot Y)+K_a(q\cdot Z).
\]
Cancelling \(K_a(q\cdot Z)\) yields \(K_a(q\cdot X)\ge K_a(q\cdot Y)\). Since \((a,q)\) was arbitrary, \(X\succeq_\kappa Y\).

The implication \(X\succeq_\kappa Y\Rightarrow X+\varepsilon e\succeq_{\cat}Y\) for every \(\varepsilon>0\) is exactly \Cref{lem:catalytic-bridge}.

It remains to prove the closure equivalence. Suppose first that \(Y\in\overline{\mathcal L_{\cat}(X)}^{\|\cdot\|_{\infty,e}}\). Then there are \(Y_n\in\mathcal L_{\cat}(X)\) with
\(\|Y_n-Y\|_{\infty,e}\to0\). For each \(n\), \(X\succeq_{\cat}Y_n\), so the first paragraph of the proof delivers \(X\succeq_\kappa Y_n\). Moreover \Cref{lem:profile-lipschitz} gives us \(\|\kappa_{Y_n}-\kappa_Y\|_\infty\to0\) and so \(\kappa_X\ge\kappa_Y\), i.e., \(X\succeq_\kappa Y\).

Conversely, suppose \(X\succeq_\kappa Y\). By \Cref{lem:catalytic-bridge}, for every \(\varepsilon>0\) there exists \(Z_\varepsilon\), independent of \((X,Y)\), such that
\[
X+\varepsilon e+Z_\varepsilon\succeq_{\mathrm{st}}Y+Z_\varepsilon.
\]
Translating closed \(C\)-upper sets by \(\varepsilon e\), this is equivalent to
\[
X+Z_\varepsilon\succeq_{\mathrm{st}}Y-\varepsilon e+Z_\varepsilon.
\]
Thus, \(X\succeq_{\cat}Y-\varepsilon e\), so \(Y-\varepsilon e\in\mathcal L_{\cat}(X)\). Since
\(\|\left(Y-\varepsilon e\right)-Y\|_{\infty,e}=\varepsilon\), we have \(Y\in\overline{\mathcal L_{\cat}(X)}^{\|\cdot\|_{\infty,e}}\).

The final equivalence follows from the closure equivalence ``\(X\succeq_\kappa Y\) \(\iff\) \(Y\in\overline{\mathcal L_{\cat}(X)}^{\|\cdot\|_{\infty,e}}\)'' and \Cref{thm:profile-order}.
\end{proof}

\bibliography{sample}

\end{document}